\documentclass[11pt]{article}
\pdfoutput=1

\usepackage{authblk}
\usepackage[margin=1in]{geometry}
\usepackage[T1]{fontenc}
\usepackage{mathpazo}
\usepackage{eucal}
\usepackage{dsfont}
\usepackage{amsmath}
\usepackage{amssymb}
\usepackage{amsthm}
\usepackage{framed}
\usepackage{mathtools}
\usepackage{authblk}
\usepackage{tikz}
\usetikzlibrary{arrows.meta,calc,shapes,positioning}
\usepackage{xspace}
\usepackage{longfbox}
\usepackage{enumitem}
\usepackage{ytableau}

\usepackage{pifont}
\usepackage{pdfpages}
\usepackage{physics}
\usepackage{hyperref}
\hypersetup{
  pdfpagemode=UseNone,
  colorlinks=true,
  citecolor=blue,
  linkcolor=blue,
  urlcolor=blue
}

\usepackage{qcircuit}
\usepackage{ytableau}
\definecolor{lightgray}{gray}{0.8}

\usepackage{cleveref}

\newtheorem{theorem}{Theorem}
\newtheorem*{theorem*}{Theorem}
\newtheorem{lemma}[theorem]{Lemma}
\newtheorem*{lemma*}{Lemma}
\newtheorem{claim}[theorem]{Claim}

\newtheorem{corollary}[theorem]{Corollary}
\newtheorem*{corollary*}{Corollary}

\theoremstyle{definition}
\newtheorem{definition}[theorem]{Definition}
\newtheorem{remark}[theorem]{Remark}

\newtheorem{algorithm}{Algorithm}

\newcommand{\smallket}[1]{| #1 \rangle}

\newcommand\C{\mathcal{C}}

\newcommand{\negl}{\mathrm{negl}}
                 
\definecolor{White}{rgb}{1,1,1}
\definecolor{Black}{rgb}{0,0,0}
\definecolor{LightGray}{rgb}{.81,.81,.81}
\colorlet{ChannelColor}{LightGray}
\colorlet{ChannelTextColor}{Black}
\colorlet{ReadoutColor}{White}

\newcommand{\Span}{\mathrm{Span}}
\newcommand{\QNC}{\mathsf{QNC}}

\usepackage[
    backend=biber,
    style=alphabetic,
    sorting=nyt,
    backref=true,
    maxbibnames=99
    ]{biblatex}
\DefineBibliographyStrings{english}{backrefpage = {page},backrefpages = {pages},mathesis = {MSc thesis},} \title{Anti-symmetrization is cheap and useful}

\author{Adam Bene Watts\thanks{University of Calgary. adam.benewatts1@ucalgary.ca}\qquad
John Bostanci\thanks{UC Berkeley \& Simons Institute. cbostanc@berkeley.edu}\qquad
Matthias C. Caro\thanks{University of Warwick. matthias.caro@warwick.ac.uk}\protect\\[0.6em]
Daniel Grier\thanks{University of California San Diego. dgrier@ucsd.edu}\qquad
Sihan Liu\thanks{University of California San Diego. sil046@ucsd.edu}\qquad
Gaurav Mahajan\thanks{Yale University and Toulouse School of Economics. gaurav.mahajan@yale.edu}}

\newcommand{\Uschur}{U_{\text{schur}}}

\date{}

\begin{document}
\maketitle
\vspace{-2em}

\begin{abstract}
We consider the computational task of anti-symmetrization of an unknown mixed quantum state: Given iid copies of a rank-$r$ qudit state $\rho=\sum_{i=1}^r \lambda_i \ket{\varphi_i}\bra{\varphi_i}$, prepare a copy of the state proportional to $\Pi_{\mathrm{anti}}^r (\ket{\varphi_1}\otimes \ldots\otimes\ket{\varphi_r})$, with $\Pi_{\mathrm{anti}}^r$ the anti-symmetric subspace projector.  
We first give a simple algorithm for anti-symmetrization that uses $\widetilde{O}(r / \lambda_{\min})$ copies of $\rho$, where $\lambda_{\min}$ is the minimum non-zero eigenvalue of $\rho$. This algorithm works even in the streaming setting with small working memory, and it uses only simple quantum operations like controlled-SWAPs and constantly many layers of single-qubit gates.
As an application of this streaming anti-symmetrization procedure, we give a reduction from gapped mixed-state one-way state generators to pure-state one-way state generators.
Then, using tools from Schur-Weyl duality, we give and analyze a second (non-streaming) anti-symmetrization procedure that uses $O(r / \lambda_{\min})$ copies of $\rho$, which then serves as a building block in new algorithms for weak Schur sampling and unitary Schur sampling, outperforming existing approaches in some regimes of number of copies $n$ and local dimension $d$, as well as in an algorithm for optimal purity amplification for qudit states.
\end{abstract}

\section{Introduction}

Mixedness is a fundamental property in quantum computation and information.  
In quantum learning theory, algorithms for tomography of mixed quantum states \cite{odonnell2016efficient, haah2017sample-optimal} have, until recently~\cite{pelecanos2025mixed}, taken on a form quite different from that of tomography algorithms for pure states \cite{bruss1999optimal}.  In quantum cryptography, various mixed-state versions of cryptographic primitives, like one-way state generators, seem to have fundamentally different hardness than their pure-state analogs~\cite{bostanci2025oracle,behera2024new,batra2024commitments}.  At the same time, at least conceptually, mixed quantum objects can almost always be purified into corresponding pure objects, for instance via purification for quantum states,  Stinespring's dilation theorem for quantum channels, or Naimark's dilation theorem for positive operator-valued measurements.  Informally, this idea has taken on the moniker ``the church of the larger Hilbert space''\footnote{This expression is attributed to John Smolin in \cite{gottesman2000quantum}.}.  The quest to understand the nature of mixedness in quantum computation has recently led to the discovery of the random purification channel~\cite{tang2026conjugate, ananth2026less, chen2026optimaltomographybosonicfermionic}, which combines the church of the larger Hilbert space with the church of the symmetric subspace \cite{harrow2013churchsymmetricsubspace}: many copies of a mixed state can be mapped to many copies of a \emph{random} purification of that mixed state.  

While the symmetric subspace has many interesting properties, its sibling, the anti-symmetric subspace, has received less attention and is sometimes overlooked. 
In this work, we present a different perspective on purifying mixed states, namely a sense in which quantum states can be canonically mapped into an \emph{anti-symmetric} pure state.  Switching from producing a symmetric state to an anti-symmetric state has certain (admittedly modest) advantages over the random purification channel, most notably that, when the number of anti-symmetrized registers reaches the rank of the mixed state, our algorithm deterministically produces a \emph{unique} pure state---the unique state in the anti-symmetric subspace of $r$ copies of an $r$-dimensional space.
The pure output state being unique and the same across all rank-$r$ states with the same support is what allows us to circumvent no-go results about universal purification \cite{liu2026universalpurificationquantummechanics}.

Our procedure also has a few tangible algorithmic benefits relative to works on random purification. For example, it can be implemented efficiently using simple tools like controlled-SWAP operations in quadratic time in the number of copies, as opposed to requiring complicated represen\-tation-theoretic transforms like the Schur transform. Additionally, one of our anti-symmetrization algorithms can be naturally implemented in the ``streaming'' model, where copies of the unknown mixed state are fed into the algorithm one copy at a time.  

\subsection{Results and discussion}

\paragraph{Anti-symmetrization algorithms.}
Our work revolves around the following anti-symmetrization problem: Given iid copies of an unknown rank-$r$ mixed qudit state $\rho$ with spectral decomposition $\rho = \sum_{i=1}^r \lambda_i \ket{\varphi_i}\bra{\varphi_i}$, prepare a copy of the pure state proportional to $\Pi_{\mathrm{anti}}^r (\ket{\varphi_1}\otimes \ldots\otimes\ket{\varphi_r})$. Here, $\Pi_{\mathrm{anti}}^r$ is the projector onto the anti-symmetric subspace of $(\mathbb{C}^d)^{\otimes r}$.
We refer to the pure state $\propto \Pi_{\mathrm{anti}}^r (\ket{\varphi_1}\otimes \ldots\otimes\ket{\varphi_r})$ as the anti-symmetric state over the support of $\rho$.

We present two algorithms for solving this anti-symmetrization problem.  
Both are based on a simple ``growth'' measurement---provided the first $k$ registers of our state are already in the anti-symmetric subspace, we measure the anti-symmetric subspace projector on the first $k+1$ registers---that allows us to iteratively increase the number of registers in the anti-symmetric subspace one register at a time.  
Clearly, if this measurement succeeds (accepts), the resulting state is then in the anti-symmetric subspace over $k+1$ registers.
To treat the reject case, we exploit a simple but crucial fact: even when the anti-symmetric subspace projector measurement on $k+1$ registers rejects (fails), we retain an anti-symmetric state in the initial $k$ registers. This allows us, in the failure case, to simply ignore the $(k+1)$st register and carry on growing the anti-symmetric state using fresh copies. 
The following theorem provides an informal statement of the guarantees for our anti-symmetric growth measurement.

\begin{theorem*}[Growing anti-symmetric states, informally]
    There is a measurement on $k+1$ registers that, if the first $k$ registers are in the anti-symmetric subspace, projects onto the support of $\Pi_{\mathrm{anti}}^{k+1}$ when it succeeds, and when it fails, outputs a state in the support of $\Pi_{\mathrm{anti}}^{k}$ on the first $k$ registers.  
    Furthermore, such a measurement can be implemented using two controlled-SWAP operations with a $\lceil\log(k+1)\rceil$-qubit control register, a single multi-controlled-$Z$ on $\lceil\log(k+1)\rceil$ qubits, and two layers of single-qubit gates on $\lceil\log(k+1)\rceil$ qubits.  
\end{theorem*}

As an aside, we show the above algorithm can be implemented in constant depth (i.e., in the class $\mathsf{QNC}^{0}_{f}$), possibly making it more practical for near term quantum computers.

We use the growth measurement as the core building block in our two algorithms for anti-symmetrization. The first of our two algorithms uses a small quantum memory of $O( r\log(d))$ qubits, by simply discarding copies of the unknown state on which the growth measurement failed. This, however, is wasteful from a sample complexity perspective. Our second algorithm uses a larger quantum memory but saves a factor of $\log(r)$ in the sample complexity, by re-using copies on which the growth measurement failed in earlier rounds to build the largest-possible anti-symmetric state from the available copies. We first state the guarantee achieved by our low-memory procedure.  

\begin{theorem*}[Anti-symmetrization without copy re-use, informally]
    There is a ``streaming-like'' algorithm that, given copies of an unknown rank-$r$ state $\rho$ with smallest non-zero eigenvalue $\lambda_{\min}$, uses at most $O(r\log(d))$ qubits of working memory, requests a total of $n = O(r \log(r) / \lambda_{\min})$ copies of $\rho$, and prepares the unique anti-symmetric state over the support of $\rho$ with success probability $0.99$.  
\end{theorem*}

The algorithm is ``streaming-like'' in the sense that it only ever acts on one fresh copy of $\rho$ at a time.
It keeps its working memory low by discarding registers that fail the growth measurement.
The analysis of this algorithm straightforwardly reduces to bounding the success probability of a coupon collector.  
However, it is in fact possible to ``re-use'' registers that failed the growth measurement in earlier rounds. This leads to our second algorithm for anti-symmetrization.

\begin{theorem*}[Anti-symmetrization with copy re-use, informally]
    There is an algorithm that, given copies of an unknown rank-$r$ state $\rho$ with smallest non-zero eigenvalue $\lambda_{\min}$, requests $n = O(r / \lambda_{\min})$ copies $\rho$ and prepares the unique anti-symmetric state over the support of $\rho$ with success probability $0.99$. 
\end{theorem*}

This second anti-symmetrization algorithm achieves its improved sample complexity by keeping the copies on which the growth measurement failed in a big pool of states and sometimes, instead of processing a fresh copy, re-applying the growth measurement on a register from that pool.
The analysis of the algorithm's success probability is more involved and requires tools from the representation theory of the symmetric group (specifically, Pieri's rule).
It also more clearly reveals the underlying advantage of the anti-symmetrization with re-use vs the one that discards failed copies: The former builds the largest representation-theoretically possible anti-symmetric state from the available $n$ copies, no matter the value of $n$; the latter does so only for sufficiently large $n$. This feature is crucial when applying our anti-symmetrization with copy re-use to develop Schur sampling routines.

\paragraph{From anti-symmetrization to Schur sampling and purity amplification.}
We apply anti-symme\-tri\-za\-tion with copy re-use to give alternative algorithms for weak and unitary Schur sampling as well as for purity amplification of qudit states.  
To describe Schur sampling, recall from Schur-Weyl duality that  
$\left(\mathbb{C}^d\right)^{\otimes n} \simeq \bigoplus_{\lambda \,\vdash_d \,n} V_\lambda \otimes W_\lambda$,
where each $W_\lambda$ is an irreducible module of $\operatorname{GL}_d$ and each $V_\lambda$ is an irreducible module of $S_n$. Here, the index $\lambda$ runs over partitions of $n$ into (at most) $d$ parts. Fixing a Schur transform $\Uschur: 
\left(\mathbb{C}^d\right)^{\otimes n} \mapsto \bigoplus_{\lambda \,\vdash_d \,n} V_\lambda \otimes W_\lambda
$, we call $\Uschur^\dagger(V_\lambda \otimes W_\lambda)$ the $\lambda$-isotypic component, and we denote the projector onto it by $\Pi_\lambda$. Then, $\delta$-approximate weak Schur sampling is the following task: Given an input state $\rho\in\mathcal{L}\left(\left(\mathbb{C}^d\right)^{\otimes n}\right)$, sample a partition $\lambda$ from a probability distribution $(p_\lambda)_{\lambda \,\vdash_d \,n}$ such that $\sum_{\lambda \, \vdash \,n} \left|p_\lambda - \text{Tr}[\Pi_\lambda \rho]\right| \leq \delta$.
In approximate unitary Schur sampling, in addition to sampling a partition $\lambda$ as in approximate weak Schur sampling, we require the algorithm to also output a state that approximately recovers the reduced state on the $W_\lambda$ module (see \Cref{subsec:schur-sampling-purification} for a detailed formulation). We use our anti-symmetrization with copy re-use to show:

\begin{corollary*}[Weak and unitary Schur sampling from anti-symmetrization, informally]
    Using anti-symmetrization, we can perform weak and unitary Schur sampling to accuracy $\delta$ with $\widetilde{O}(n^2 d\log(1/\delta))$ controlled-SWAP operations and $T$-count $\widetilde{O}(n^2 d\log^2(1/\delta))$.
\end{corollary*}

Prior approaches to constant-accuracy weak and unitary Schur sampling for general $d$, either using the full Schur transform \cite{bacon2005quantum, bacon2006efficientquantumcircuitsschur, Krovi2019efficienthigh, burchardt2025highdimensionalquantumschurtransforms} or aiming directly at weak/unitary Schur sampling via Clebsch-Gordan transforms \cite{cervero2023weakschursamplinglogarithmic}, have gate complexities of the form $\widetilde{O}(n^a d^b)$ with integers $a\geq 1$ and $b\geq 0$ satisfying $a+b\geq 4$. Our procedure achieves $a=2$ and $b=1$, so $a+b=3$, thus improving the overall complexity in some regimes of $n$ and $d$. For $d=2$, \cite{brahmachari2025optimal}'s unitary Schur sampling algorithm using $\widetilde{O}(n)$ random SWAP tests outperforms ours.

Motivated by the need to combat noise in quantum systems, Li et al.~\cite{li2024optimal} introduced the task of purity amplification: Given $n$ iid copies of a mixed qudit state $\rho$, output a state $\rho'$ that is close to achieving the maximum overlap with the top eigenvector of $\rho$. We develop an anti-symmetrization-based procedure that serves as a simplification of the optimal quantum purity amplification protocol of \cite{li2024optimal} while producing the same output state.

\begin{corollary*}[Purity amplification from anti-symmetrization, informally]
    Using anti-symmetrization, we can achieve optimal purity amplification to accuracy $\delta$ with $\widetilde{O}(n^2 d\log(1/\delta))$ controlled-SWAP operations and $T$-count $\widetilde{O}(n^2 d\log^2(1/\delta))$.
\end{corollary*}

Prior work achieved quantum purity amplification with an overall gate complexity polynomial in $n$ and $\log (d)$, using generalized quantum phase estimation \cite{li2024optimal, li2026quantumpurityamplificationarbitrary}. While our dependence on the local dimension $d$ is worse, if $d$ is taken to be fixed, then we achieve an improved gate complexity with an arguably simpler implementation.  
We note that, as our purity amplification procedure builds on our anti-symmetrization with copy re-use, it does not apply to the streaming version of the task \cite{Childs2025streamingquantum, grier2025streamingquantumstatepurification}.

\paragraph{Anti-symmetrization for one-way state generators.}
Our second application is in quantum cryptography.  In their original work, \cite{morimae2024one} defined a one-way state generator (OWSG) to be a keyed family of states $\phi_{k}$ such that it is hard to learn the key $k$ from copies of $\phi_{k}$ but computationally easy to verify $\phi_k$ given $k$.  
If one requires $\phi_{k}$ to be pure, there are at least two natural verification algorithms: one may either perform repeated SWAP tests between copies of $\phi_k$ and the unknown state, or one may uncompute the state preparation circuit for $\phi_k$ on the unknown state and measure in the computational basis, checking for an all-zero outcome.
When $\phi_{k}$ is not pure, these strategies do not immediately determine whether an unknown state is equal to $\phi_{k}$. However, there may be other verification strategies for low-rank $\phi_k$. 
Prior work left open the question of whether pure-state output OWSGs were a different kind of hardness assumption than mixed-state output OWSGs, even when the rank of those mixed states is low.  We present a simple reduction from low-rank, gapped mixed-state output OWSGs to pure-state output OWSGs.  

\begin{corollary*}
    If there is a low-rank one-way state generator where every $\phi_{k}$ has a minimum eigenvalue that is inverse-polynomial in the security parameter $\lambda$, then pure-state one-way state generators exist.
\end{corollary*}

This reduction is achieved by applying either of our two anti-symmetrization algorithms to a low-rank, gapped mixed-state output OWSG, exploiting that, for each $k$, there is a unique pure state obtained as the output of anti-symmetrization. 
Furthermore, since our anti-symmetrization algorithms are efficient, computational one-wayness carries over to the anti-symmetrized state.  
We also show that there is likely no black-box reduction (that only uses copies of the state) from mixed-state OWSGs to pure-state OWSGs beyond $\mathrm{poly}(\lambda)$-rank mixed states by extending the oracle separation of \cite{bostanci2025oracle} to $\lambda^{\log^{1 + \epsilon} (\lambda)}$-rank OWSGs.  

\subsection{Open questions}

\paragraph{Optimal bounds for (streaming) anti-symmetrization.}
We achieved streaming anti-symmetriza\-tion from $O(r\log(r)/\lambda_{\mathrm{min}})$ copies. Without the streaming requirement we showed that $O(r/\lambda_{\mathrm{min}})$ copies suffice. This raises two natural questions: First, what is the optimal complexity of anti-symmetrization? 
Among protocols that only act through permutation operations, arguments from Schur-Weyl duality suggest that $\Theta(\sum_{i=1}^r \lambda_i^{-1})$ copies are necessary and sufficient. But among general protocols, only $\Omega(\lambda_{\mathrm{min}}^{-1})$ copies may be needed. Is there a protocol matching this bound? Answering this question will also help clarify whether a reduction from gapped low-rank mixed-state output OWSGs to their pure-state counterpart via anti-symmetrization fundamentally requires the gappedness assumption. Second, can the optimal complexity of anti-symmetrization be achieved by a streaming algorithm that uses a small (say, polynomial in $r$ and $\log(d)$) working memory, or are there trade-offs between the available quantum memory and the achievable sample complexity?

\paragraph{Anti-symmetrization for support learning and testing.} The random purification channel of \cite{tang2026conjugate} and its variants have proven useful in developing algorithms for learning mixed states via existing pure-state procedures, see e.g. \cite{pelecanos2025mixed, mele2026optimallearningquantumchannels, chen2026quantumchanneltomographyoptimal, chen2026optimaltomographybosonicfermionic, girardi2026randomstinespringsuperchannelconverting}. 
However, in the context of property testing, using the random purification channel is more challenging because any single mixed state admits a whole (local-unitary) orbit of possible purifications, from which it draws one at random. Thus, even if the initial testing problem is about certifying a single mixed target hypothesis state, under random purification it becomes a task of certifying a target orbit of pure states. Anti-symmetrization with its unique pure target state may help here. Note that anti-symmetrization applied to copies of an unknown mixed state as studied in this work intentionally loses much of the spectral information (except the number of non-zero eigenvalues), but it retains information about the state's support. It therefore may be used to reduce tasks of learning or testing properties of the support of a mixed state to corresponding pure anti-symmetric state problems.

\paragraph{Anti-symmetrization for quantum chemistry.}
In quantum chemistry, pure states in the anti-symmetric subspace are often referred to as \emph{Slater determinants}, which describe certain non-interact\-ing fermionic systems. They appear as central objects in many mean-field theories, such as Hartree-Fock theory and density functional theory \cite{cances2023density}, are important to our understanding of chemical processes~\cite{chan2024quantumchemistryclassicalheuristics}, and serve as a starting point for simulating fermionic interactions \cite{fomichev2024initial}.   
Whereas standard algorithms for preparing a Slater determinant start from a classical description, which comes either from a theoretical prediction or an experimental setup, in a world with quantum computers, quantum systems may instead be specified implicitly through copies of quantum states.
Does our streaming anti-symmetrization, which can prepare a pure Slater determinant from such a quantum description without having to learn a full description of the state, be used for applications in quantum chemistry?

\paragraph{Note added.} While we were preparing our manuscript, the closely related, independent work \cite{brahmachari2026unitaryschursamplingqudits} by Shrigyan Brahmachari, David Jakab, Henry D.~Pfister, and Iman Marvian appeared on the arXiv. An earlier work that involved many of the same authors~\cite{brahmachari2025optimal} was a major influence on~\Cref{sec:schur_and_purification} of our paper.
The main goal of \cite{brahmachari2026unitaryschursamplingqudits} is unitary Schur sampling, which that work achieves via the same growth routine as our \Cref{alg:anti-sym-growing} to iteratively extract the largest anti-symmetric subsystem. Thus, the relationship between Schur sampling and anti-symmetrization is common to our works, both conceptually and technically.  While we implement unitary Schur sampling and purity amplification for qudits with $\widetilde{O}(n^2d)$ controlled-SWAP operations, their implementation uses $\widetilde{O}(\min\{n^4, n d^3\})$ pairwise SWAP tests. Our proofs were completed before~\cite{brahmachari2026unitaryschursamplingqudits} appeared on the arXiv, but (as we acknowledge below) they were developed with the help of AI and we cannot rule out the possibility that this AI was influenced by an unpublished draft of their work. 

\section*{Acknowledgments and AI disclosure}
We thank Jack Spilecki for helpful discussions about representation theory and applications of the random purification channel.

AI made a substantive contribution to this work. The human authors wrote an early draft of this work, including algorithms for anti-symmetrization and symmetrization that involved substantially more representation theory.  When asked to review said draft, ChatGPT-6 Pro proposed simplifications to the earlier anti-symmetrization algorithm that the authors had developed, which led to the current presentation of \Cref{alg:anti-sym-growing,alg:sample-inefficient-antisym}. Claude Opus 5.5 and Fable 5.1 were used to aid with technical aspects of representation theory in developing \Cref{sec:schur_and_purification} and to help improve presentation of the mathematical arguments.
All final proofs were written and independently verified by the authors. \section{Preliminaries}
\label{sec:math_prelims}

We consider a system of $n$ qudits, described by the Hilbert space $\left(\mathbb{C}^d\right)^{\otimes n}$. We let $S_n$ denote the symmetric group on $n$ elements, and for any permutation $\sigma \in S_n$ we let $P_\sigma$ denote the natural action of $\sigma$ on this space, defined by
\begin{align}
    P_\sigma \left( \ket{\psi_1} \otimes \ket{\psi_2}\otimes ... \otimes \ket{\psi_n}\right) =  \smallket{\psi_{\sigma^{-1}(1)}} \otimes \smallket{\psi_{\sigma^{-1}(2)}} \otimes ... \otimes \smallket{\psi_{\sigma^{-1}(n)}}.
\end{align}
When a permutation $\tau$ acts on fewer than $n$ elements, we define  $P_\tau$ to act on the first registers of the space. So, for example, for a permutation $\tau \in S_{n-1}$,
\begin{align}
P_\tau \left( \ket{\psi_1} \otimes \ket{\psi_2}\otimes ... \otimes \ket{\psi_{n-1}}\otimes \ket{\psi_n}\right) =  \smallket{\psi_{\tau^{-1}(1)}} \otimes \smallket{\psi_{\tau^{-1}(2)}}\otimes ... \otimes \smallket{\psi_{\tau^{-1}(n-1)}}\otimes \smallket{\psi_n}.
\end{align}

The anti-symmetric subspace projector on the first $k$ qudits of this space is denoted by $\Pi_{\mathrm{anti}}^k$ and defined as
\begin{align}
    \Pi_{\mathrm{anti}}^k := \frac{1}{k!}\sum_{\sigma \in S_k} \text{sgn} (\sigma)P_\sigma.
\end{align}
Routine calculation shows that this is indeed a self-adjoint projector.

We now make a standard but useful observation: the anti-symmetric subspace projector on $k+1$ registers can be written in terms of the anti-symmetric subspace projector on $k$ registers multiplied by a sum of transpositions. We formalize this in the following lemma. 

\newcommand{\transposition}[2]{({#1}, {#2})}
\begin{lemma}
\label{lem:anti_decomp}
    We have 
\begin{align*}
    \left(1 - \sum_{i=1}^k P_{\transposition{i}{k+1}} \right)\Pi_{\mathrm{anti}}^k = \Pi_{\mathrm{anti}}^k  \left(1 - \sum_{i=1}^k P_{\transposition{i}{k+1}} \right) = (k+1)\Pi_{\mathrm{anti}}^{k+1} \, .
\end{align*}
\end{lemma}

\begin{proof}
    Multiplying out terms gives that 
    \begin{align*}
    \left(1 - \sum_{i=1}^k P_{\transposition{i}{k+1}} \right)\Pi_{\mathrm{anti}}^k &= \frac{1}{k!}\left(1 - \sum_{i=1}^k P_{\transposition{i}{k+1}} \right) \sum_{\tau \in S_k} \text{sgn}(\tau) P_\tau \\
    &= \frac{1}{k!} \left( \sum_{\sigma\in S_{k+1}}  \text{sgn}(\sigma) P_\sigma \right)\\
    &= (k+1)\Pi_{\mathrm{anti}}^{k+1}\, ,
    \end{align*}
    where the equality on the second line follows by noting that any element in $S_{k+1}$ can be written uniquely as an element in $S_k$ left multiplied by either a transposition touching $k+1$ or the identity. (To verify this, simply write a permutation $\sigma \in S_{k+1}$ as $\sigma = (k+1,  \sigma(k+1))^2 \circ \sigma$ and then note $(k+1,  \sigma(k+1)) \circ \sigma \in S_k$.) A similar proof yields the other equality claimed in the lemma, only now multiplying by transpositions on the right. 
\end{proof}

Motivated by~\Cref{lem:anti_decomp}, we define the operator 
\begin{align}\label{eq:def-average-swap-operator}
    A^{[k],k+1} := \frac{1}{k+1}\left(1 - \sum_{i=1}^k P_{\transposition{i}{k+1}}\right).
\end{align}
We also note, as a simple corollary of~\Cref{lem:anti_decomp}, that anti-symmetric subspace projectors commute. 
\begin{corollary}\label{cor:proj_commute}
    We have
    \begin{align}
    \Pi^{k+1}_{\mathrm{anti}}
    \Pi^k_{\mathrm{anti}} = 
    \Pi^k_{\mathrm{anti}}\Pi^{k+1}_{\mathrm{anti}}
    \, .
    \end{align}
\end{corollary}
\begin{proof}
\Cref{lem:anti_decomp} gives
\begin{align}
    \Pi^{k}_{\mathrm{anti}}\Pi^{k+1}_{\mathrm{anti}} = \Pi^k_{\mathrm{anti}}A^{[k], k+1}\Pi^k_{\mathrm{anti}} = \Pi^{k+1}_{\mathrm{anti}}\Pi^k_{\mathrm{anti}}
\end{align}
\end{proof}

We extend the notation developed so far to act on arbitrary subsets of $n$ qudits. For any $S \subseteq [n]$ let $\Pi^S_{\mathrm{anti}}$ be the anti-symmetric subspace projector acting on those qudits. Similarly, for every $S \subseteq [n]$ and $j \in [n] \backslash S$ define 
\begin{align}
A^{S, j} = \frac{1}{|S| + 1} \left(1 -\sum_{s \in S} P_{\transposition{s}{j}}\right).
\end{align}
Note that for any such $S$ and $j$ we have 
\begin{align}
A^{S, j} \Pi^S_{\mathrm{anti}} = \Pi^S_{\mathrm{anti}} A^{S, j}  = \Pi^{S\cup\{j\}}_{\mathrm{anti}}
\end{align}
and 
\begin{align}
\Pi^S_{\mathrm{anti}} \; \Pi^{S\cup\{j\}}_{\mathrm{anti}} = \Pi^{S\cup\{j\}}_{\mathrm{anti}} \Pi^S_{\mathrm{anti}}.
\end{align} \section{Two algorithms for anti-symmetrization}
\label{sec:antisym_sec}

This section gives a formal description and analysis of two algorithms for preparing the anti-symmetric state over the support of an unknown mixed state from access to a stream of copies. It begins with~\Cref{alg:anti-sym-growing}, which describes an efficient circuit that implements the two-outcome ``growth'' measurement
\begin{align}
    \left\{ A^{[k], k+1}, \left(1- A^{[k], k+1}\right)^{1/2}\right\}
\end{align}
on the subspace of all states $\ket{\psi}$ such that $\Pi^k_{\mathrm{anti}}\ket{\psi}=\ket{\psi}$.
We call the first outcome of this measurement the success outcome, and the second the failure outcome.
When this measurement is applied to an anti-symmetric state on the first $k$ registers and succeeds, \Cref{lem:anti_decomp} implies that it produces an anti-symmetric state on $k+1$ registers. When it fails, it still produces an anti-symmetric state on the first $k$ registers because of  \Cref{cor:proj_commute}.

\Cref{alg:sample-inefficient-antisym} repeatedly applies the measurement described in~\Cref{alg:anti-sym-growing} to the anti-symmetric state grown so far and a fresh copy taken from a stream of iid mixed-state copies to produce an anti-symmetric state over the support of the unknown state. When the growth measurement succeeds, the anti-symmetric state grows by one register. When it fails, the mixed-state register is discarded. In either case, the process repeats. 

\Cref{alg:sample-efficient-antisym} is a more sample-efficient but less time-efficient version of~\Cref{alg:sample-inefficient-antisym} in which, rather than being discarded, mixed states involved in a failure are saved and retested later. The representation-theoretic properties of this algorithm are discussed in~\Cref{sec:schur_and_purification}.

We now give a formal description of each of these algorithms and their analysis. For ease of presentation, we write \Cref{alg:anti-sym-growing} as processing a pure input state, but it can clearly also be invoked on mixed states, and the guarantees extend accordingly by linearity.

\begin{longfbox}[breakable=true, padding=1em, margin-top=1em, margin-bottom=1em]
\begin{algorithm}
Anti-symmetrization growth measurement
\label{alg:anti-sym-growing}
\end{algorithm}
\noindent \textbf{Input:} A state $\ket{\psi} \in \left(\mathbb{C}^d\right)^{\otimes k+1}$ satisfying $\Pi^k_{\mathrm{anti}}\ket{\psi} = \ket{\psi}$.

\noindent \textbf{Output:} Either: \textit{Success}, and a state $\ket{\phi} \in \left(\mathbb{C}^d\right)^{\otimes k+1}$ satisfying $\Pi^{k
+1}_{\mathrm{anti}}\ket{\phi} = \ket{\phi}$; or \textit{Failure}, and a state $\ket{\psi'} \in \left(\mathbb{C}^d\right)^{\otimes k+1}$ satisfying $\Pi^{k
}_{\mathrm{anti}}\ket{\psi'} = \ket{\psi'}$.
\begin{enumerate}
    \item Let $W$ be any unitary preparing the uniform superposition on a $(k+1)$-dimensional ancilla register: $W \ket{0} = \frac{1}{\sqrt{k+1}} \sum_{i=0}^k\ket{i}$. Prepare this state, and label it the \texttt{ancilla} register. Call the register holding $\ket{\psi}$ the \texttt{active} register.  
    \item For each $ 1 \leq i \leq k$, apply the (negated) swap operation $-P_{(i, k+1)}$ to the state $\ket{\psi}$ controlled on the ancilla being in state $\ket{i}$. That is, apply the unitary 
    \begin{align}
        U_{A} = \ketbra{0}{0} \otimes I - \sum_{i=1}^k \ketbra{i}{i} \otimes P_{(i, k+1)} 
    \end{align}
    to the \texttt{ancilla} and \texttt{active} registers.
    \item Apply $W^\dag$.
\item Apply the two-outcome projective measurement 
    \begin{align}
        \left\{ \ketbra{0}{0}, 1- \ketbra{0}{0} \right\}
    \end{align}
    to the \texttt{ancilla} register. 
\item If outcome $\ketbra{0}{0}$ is observed, report \textit{success} and return the state in the \texttt{active} register. 
    \item If outcome $1 - \ketbra{0}{0}$ is observed, report \textit{failure}, apply the operator $U_A (W \otimes I)$  to the active and ancilla registers, then return the state in the \texttt{active} register. 
\end{enumerate}
\end{longfbox}

\begin{theorem}
    \Cref{alg:anti-sym-growing} is correct, i.e., it outputs a state with the claimed anti-symmetry properties in both the \textit{success} and \textit{failure} cases. Furthermore, the algorithm outputs \textit{success} with probability $\bra{\psi}\Pi^{k+1}_{\mathrm{anti}} \ket{\psi}$.  
\end{theorem}
\begin{proof}
    We analyze the \textit{success} case first. This can be understood as an example of the linear combination of unitaries technique. After the $\ketbra{0}{0}$ outcome is observed in step 5 of the algorithm, the (unnormalized) state in the \texttt{ancilla} and \texttt{active} registers is: 
    \begin{align}
        &\frac{1}{\sqrt{k+1}}\left(\ketbra{0}{0} \otimes I\right)(W^\dag \otimes I) U_A \left( \sum_{i=0}^k \ket{i} \otimes \ket{\psi} \right)\\ &\hspace{40pt}= \frac{1}{\sqrt{k+1}}\left(\ketbra{0}{0} \otimes I\right)(W^\dag \otimes I)\left( \ket{0}\ket{\psi} - \sum_{i=1}^k \ket{i} P_{(i, k+1)} \ket{\psi} \right) \\
        &\hspace{40pt}=\frac{1}{k+1} \left(\ket{0}\left(\sum_{i=0}^k  \bra{i}\right) \otimes I \right) \left( \ket{0}\ket{\psi} - \sum_{i=1}^k \ket{i} P_{(i, k+1)} \ket{\psi} \right)\\
        &\hspace{40pt}= \frac{1}{k+1} \ket{0} \otimes \left(1 - \sum_{i=1}^k  P_{(i,k+1)} \right) \ket{\psi} \\
        &\hspace{40pt}= \ket{0} \otimes A^{[k], k+1} \ket{\psi},
    \end{align}
    where we recall the definition of $A^{[k], k+1}$ given in~\Cref{eq:def-average-swap-operator}. By assumption of the algorithm, we have $\ket{\psi} = \Pi^k_{\mathrm{anti}}\ket{\psi}$ and therefore
    \begin{align}
        A^{[k], k+1} \ket{\psi} &= A^{[k], k+1} \Pi^k_{\mathrm{anti}}\ket{\psi} \\
        &= \Pi^{k+1}_{\mathrm{anti}}\ket{\psi} \, ,
    \end{align}
    which both gives the probability of the algorithm outputting \textit{success}, and proves correctness of the algorithm in this case.

    In the \textit{failure} case, the overall operator applied to the \texttt{active} and \texttt{ancilla} registers is 
\begin{align}
        &U_A (W \otimes I) \left((I - \ketbra{0}{0})\otimes I \right)(W^\dag \otimes I)U_A = I - U_A \left(\left(W \ketbra{0}{0} W^\dag \right) \otimes I\right) U_A.
    \end{align}
    Building on the analysis of the \textit{success} case gives 
    \begin{align}
        &\frac{1}{\sqrt{k+1}}U_A \left(\left(W \ketbra{0}{0} W^\dag \right) \otimes I\right) U_A \left( \sum_{i=0}^k \ket{i} \otimes \ket{\psi} \right) \\
        &\hspace{40pt}= U_A \left( W \ket{0} \otimes \Pi^{k+1}_{\mathrm{anti}} \ket{\psi}\right) \\
        &\hspace{40pt}= \frac{1}{\sqrt{k+1}}\left( \ket{0} \otimes I - \sum_{i=1}^k \ket{i} \otimes P_{(i, k+1)} \right) \Pi^{k+1}_{\mathrm{anti}} \ket{\psi} \\
        &\hspace{40pt}= W \ket{0} \otimes \Pi^{k+1}_{\mathrm{anti}} \ket{\psi} 
\end{align}
    from which it follows that the state in the \texttt{ancilla} and \texttt{active} registers after applying all operations in the \textit{failure} branch of the algorithm is
    \begin{align}
        &\left(I - U_A \left(\left(W \ketbra{0}{0} W^\dag \right) \otimes I\right) U_A\right)\left( W \ket 0 \otimes \ket{\psi}\right) 
        = W \ket 0 \otimes (I- \Pi^{k+1}_{\mathrm{anti}})\ket{\psi}\, ,
    \end{align}
    and the state returned is 
    \begin{align}
        \ket{\psi'} &= (I - \Pi^{k+1}_{\mathrm{anti}}) \ket{\psi} \\
        &= (I - \Pi^{k+1}_{\mathrm{anti}}) \Pi^{k}_{\mathrm{anti}}\ket{\psi} \\
        &= \Pi^{k}_{\mathrm{anti}}(I - \Pi^{k+1}_{\mathrm{anti}})\ket{\psi} = \Pi^{k}_{\mathrm{anti}}\ket{\psi'}\, ,
    \end{align}
    where we have used the assumption on $\ket{\psi}$ on the second line, and~\Cref{cor:proj_commute} on the third. 
\end{proof}

We briefly describe a low-depth implementation of \Cref{alg:anti-sym-growing} given large fanout gates:
\begin{lemma}
    \Cref{alg:anti-sym-growing} can be implemented in $\QNC^0_f$, that is, polynomial-size, constant-depth circuits with unbounded fanout gates.
\end{lemma}
\begin{proof}
    We describe constant-depth implementations for each of the first three steps of the algorithm over qubits, though the construction would be similar for other local dimensions:

    \vspace{1em}
    \noindent\textit{Step $1$ \& $3$:} It will be easiest if we prepare the states in their ``unary'' encoding, with $\ket{0}$ represented as $\ket{0^k}$ and $\ket{i}$ represented by $\ket{0, \ldots, 1, \ldots 0}$ where the $1$ is in the $i$th position. To do this, we can first create the $W$-state on $k+1$ qubits \cite{grier2026mathsfqac0containsmathsftc0with}, flip the $(k+1)$st qubit if the first $k$ are in the all-zero state (with one generalized Toffoli gate), and then drop the last qubit. 

    \vspace{1em}
    \noindent\textit{Step $2$:} The key barrier to parallelizing this step is that all the (controlled) $P_{(i,k+1)}$ terms touch the $(k+1)$st qudit of $\ket{\psi}$. To circumvent this issue, we will use a variant of a technique from \cite{hoyer2005quantum}, which parallelizes mutually commuting operations. While the $P_{(i,k+1)}$ transpositions do not commute, the controlled versions do. Below is the sequence of operations in $\QNC^0_f$ we use to implement this scheme. It uses, at its heart, the fact that a SWAP gate can be built from three CNOT gates. We track how the classical basis states evolve under each operation. For simplicity, first assume the local dimension is $d=2$. The registers below are (in order): $k$-bit ancilla register; first $k$ bits of $\ket{\psi}$; the $(k+1)$st bit of $\ket{\psi}$ appended with $(k-1)$ zero bits; $k$ zero bits; $1$ zero bit. We have (when $i \neq 0$):
    \begin{align*}
        & \ket{i} \ket{x_1, \ldots, x_i, \ldots, x_k} \ket{y, 0, \ldots, 0} \ket{0, \ldots, 0, \ldots, 0}\ket{0} \\
        \xrightarrow{\mathsf{Fanout}} \;\;& \ket{i} \ket{x_1, \ldots, x_i, \ldots, x_k} \ket{y, \ldots, y, \ldots y} \ket{0, \ldots, 0, \ldots, 0} \ket{0} \\
        \xrightarrow{\mathsf{Toffolis}} \;\;& \ket{i} \ket{x_1, \ldots, y \oplus x_i, \ldots, x_k} \ket{y, \ldots, y, \ldots y} \ket{0, \ldots, y \oplus x_i, \ldots, 0} \ket{0} \\
        \xrightarrow{\mathsf{Parity}} \;\;& \ket{i} \ket{x_1, \ldots, y \oplus  x_i, \ldots, x_k} \ket{y, \ldots, y, \ldots y} \ket{0, \ldots, y \oplus x_i, \ldots, 0} \ket{y \oplus x_i} \\        
        \xrightarrow{\mathsf{Fanout}} \;\;& \ket{i} \ket{x_1, \ldots, y \oplus x_i, \ldots, x_k} \ket{x_i, \ldots, x_i, \ldots x_i} \ket{0, \ldots, y \oplus x_i, \ldots, 0} \ket{y \oplus x_i} \\
        \xrightarrow{\mathsf{Parity}} \;\;& \ket{i} \ket{x_1, \ldots, y \oplus x_i, \ldots, x_k} \ket{x_i, \ldots, x_i, \ldots x_i} \ket{0, \ldots, y \oplus x_i, \ldots, 0} \ket{0} \\
        \xrightarrow{\mathsf{Toffolis}} \;\;& \ket{i} \ket{x_1, \ldots, y , \ldots, x_k} \ket{x_i, \ldots, x_i, \ldots x_i} \ket{0, \ldots, 0, \ldots, 0} \ket{0} \\
        \xrightarrow{\mathsf{Fanout}} \;\;& \ket{i} \ket{x_1, \ldots, y , \ldots, x_k} \ket{x_i, \ldots, 0, \ldots 0} \ket{0, \ldots, 0, \ldots, 0} \ket{0}
    \end{align*}
    When $i = 0$, the circuit construction above does nothing as desired. To obtain the minus signs, we only need to apply a layer of $Z$ gates to our ancilla register since at most one location will be a $1$. To generalize to multiple qubits (i.e., $d = 2^m$), we simply run the entire circuit above in parallel for each bit. We use a Fanout gate to parallelize the controls even further; that is, the ancilla state $\ket{i}$ becomes $\ket{0^m, \ldots, 1^m, \ldots 0^m}$ where the $1$'s are in the $i$th block of $n$ bits.
\end{proof}

Next, we describe and analyze our first anti-symmetrization algorithm, which repeatedly applies~\Cref{alg:anti-sym-growing} to prepare an anti-symmetric state while discarding failed copies.
 
\begin{longfbox}[breakable=true, padding=1em, margin-top=1em, margin-bottom=1em]
\begin{algorithm}
Anti-symmetrization without copy re-use
\label{alg:sample-inefficient-antisym}
\end{algorithm}
\noindent \textbf{Input:} $n$ copies of some unknown mixed state $\rho \in \mathcal{L}(\mathbb{C}^d)$. \\
\noindent \textbf{Output:} An anti-symmetric state on $k$ qudits for some $1 \leq k \leq d$. 
\begin{enumerate}
    \item Label three registers: \texttt{anti-symmetric}, \texttt{input}, and \texttt{discard}. 
    \item Set $k=1$ and add one copy of $\rho$ to the \texttt{anti-symmetric} register, and all others to the \texttt{input} register. 
    \item While the \texttt{input} register is non-empty and $k < d$:
    \begin{enumerate}
        \item Take the state from the \texttt{anti-symmetric} register and the first qudit in the \texttt{input} register, then apply~\Cref{alg:anti-sym-growing} to them. 
        \item If the run of \Cref{alg:anti-sym-growing} from (a) returns \textit{success}, increase $k$ by 1, and add the state output by the test to the \texttt{anti-symmetric} register. 
        \item If the run of \Cref{alg:anti-sym-growing} from (a) returns \textit{failure}, add the first $k$ qudits of the state returned to the \texttt{antisym\allowbreak metric} register and the last qudit to the \texttt{discard} register. 
    \end{enumerate}
    \item Return the state in the \texttt{anti-symmetric} register. 
\end{enumerate}
\end{longfbox}

\begin{theorem}
    \label{thm:sample-inefficient-analysis}
    \Cref{alg:sample-inefficient-antisym} is correct, that is, it always returns an anti-symmetric state. Furthermore if the input state $\rho$ has rank $r$ and minimum non-zero eigenvalue $\lambda_{\mathrm{min}}$ then \Cref{alg:sample-inefficient-antisym} returns the unique anti-symmetric state over the support of $\rho$ with probability at least $1-r\left(1-\frac{\lambda_{\mathrm{min}}}{r} \right)^{n} \geq 1- r \exp(-n\lambda_{\mathrm{min}}/r)$. 
    
    Explicitly, if
    \begin{align}
        \rho = \sum_{i=1}^r \lambda_i \ketbra{\varphi_i}{\varphi_i}
    \end{align}
    then~\Cref{alg:sample-inefficient-antisym} returns a state proportional to
    \begin{align}
        \Pi_{\mathrm{anti}}^r \left(\ket{\varphi_1}\otimes\ket{\varphi_2} \otimes ... \otimes \ket{\varphi_r}\right)
    \end{align}
    with probability as bounded above. 
\end{theorem}

\begin{proof}
    Correctness of \Cref{alg:sample-inefficient-antisym} is immediate from correctness of~\Cref{alg:anti-sym-growing}. It remains to bound the probability that this algorithm outputs the anti-symmetric state over the support of $\rho$. 
    
    Rather than analyzing the algorithm over mixed states directly, we interpret each $\rho$ as a probabilistic mixture of orthogonal pure (eigen)states $\{\ket{\varphi_1}, \ket{\varphi_2},..., \ket{\varphi_r}\}$ and analyze the algorithm on each pure state in the resulting mixture. As the algorithm only acts on these states with controlled-SWAP operations and since \Cref{alg:anti-sym-growing} always outputs an anti-symmetric state, any state stored in the \texttt{anti-symmetric} register during the run of the algorithm must be equivalent to a mixture of states produced by applying the anti-symmetric subspace projector applied to some subset of states from $\{\ket{\varphi_1}, \ket{\varphi_2},..., \ket{\varphi_r}\}$. Consider the case where this subset has size $k$, and describe one arbitrary state in this mixture as 
    \begin{align}
        \ket{\Phi^{S}} \propto \Pi^k_{\mathrm{anti}} \left(\bigotimes_{s \in S}\ket{\varphi_s} \right)\, ,
    \end{align}
    where $S \subseteq [r]$, $|S| = k$. Now straightforward calculation gives 
    \begin{align}
        \bra{\Phi^{(S)}} \bra{\varphi_j}\Pi^{k+1}_{\mathrm{anti}} \ket{\Phi^{(S)}} \ket{\varphi_j} = 0
    \end{align}
    if $j \in S$, and 
    \begin{align}
        \bra{\Phi^{(S)}} \bra{\varphi_j}\Pi^{k+1}_{\mathrm{anti}} \ket{\Phi^{(S)}} \ket{\varphi_j} &= \bra{\Phi^{(S)}} \bra{\varphi_j}A^{[k], k+1}\Pi^{k}_{\mathrm{anti}} \ket{\Phi^{(S)}} \ket{\varphi_j} = \frac{1}{k+1}
    \end{align}
    if $j \notin S$. A fresh state drawn from the \texttt{input} register is in the latter $j\not\in S$ case with probability 
    \begin{align}
        \sum_{j \notin S} \lambda_j \geq (r-k)\lambda_{\mathrm{min}}.
    \end{align}
    So we conclude that every time a state is removed from the \texttt{input} register, the next measurement (i.e., the next run of \Cref{alg:anti-sym-growing}) increases the size of the state in the \texttt{anti-symmetric} register with probability at least 
    \begin{align}
        \frac{\lambda_{\mathrm{min}}(r-k)}{(k+1)} \geq \frac{\lambda_{\mathrm{min}}(r-k)}{r}\, ,
    \end{align}
    since the size of the state in the anti-symmetric register is strictly less than $r$ while that register is incomplete. The latter probability is exactly the probability of incrementing the number of collected coupons if there are $r$ coupons overall and each coupon appears with probability $\lambda_{\mathrm{min}}/r$.
    So, standard union bound arguments give that the probability of not collecting all coupons after $n$ rounds is upper bounded by
    \begin{align}
        \sum_{i \in [r]} \text{Pr[Coupon $i$ is not collected]} =  r\left(1-\frac{\lambda_{\mathrm{min}}}{r} \right)^{n} \leq r e^{-n\lambda_{\mathrm{min}}/r} \, .
    \end{align}
    This is the claimed bound on the probability of~\Cref{alg:sample-inefficient-antisym} returning the desired anti-symmetric state. 
\end{proof}

\begin{remark}
    The proof of \Cref{thm:sample-inefficient-analysis} also shows the following: If we are promised in advance that the unknown state $\rho$ has rank $r$, then we can change the WHILE loop in Step 3 of \Cref{alg:sample-inefficient-antisym} to be over $k<r$. If $\rho$ has rank $r$, then the algorithm can never increment the counter $k$ beyond $r$.
\end{remark}

In \Cref{alg:sample-inefficient-antisym}, states moved to the \texttt{discard} register are never used again. Thus, they can be traced out without affecting the algorithm's performance. This makes it well suited for a ``streaming'' model in which new states are added to the \texttt{input} register one at a time, and only states in the \texttt{anti-symmetric} register are stored. 

However, it is possible to improve the sample complexity of~\Cref{alg:sample-inefficient-antisym} by abandoning this streaming-like behavior. The following key observation underlies this improvement: operators applied to states in the \textit{Failure} mode case of~\Cref{alg:anti-sym-growing} across different rounds do not commute with one another, so it is possible that a state which is initially rejected by the algorithm can later be used to grow the anti-symmetric state (after a different state is rejected).
\Cref{alg:sample-efficient-antisym} describes this process formally. This algorithm also plays a key role in the quantum purification techniques described in~\Cref{sec:schur_and_purification}. Because this algorithm is used in these two different contexts, we allow its input to be an arbitrary state; the subsequent analysis in~\Cref{thm:sample-efficient-analysis} then describes the algorithm's performance in the special case that its input is $\rho^{\otimes n}$.

\begin{longfbox}[breakable=false, padding=1em, margin-top=1em, margin-bottom=1em]
\begin{algorithm}
Anti-symmetrization with copy re-use
\label{alg:sample-efficient-antisym}
\end{algorithm}
\noindent \textbf{Input:} An unknown state $\rho \in \mathcal{L}\left(\left(\mathbb{C}^d\right)^{\otimes n}\right)$ and a confidence parameter $\epsilon$. \\
\noindent \textbf{Output:} An anti-symmetric state on $k$ qudits for some $1 \leq k \leq d$. 
\begin{enumerate}
    \item Label two registers (sets of qudits): \texttt{anti-symmetric} and \texttt{candidates}. 
    \item Set $k=1$ and the first qudit of the input state to the \texttt{anti-symmetric} register, and all other qudits to the \texttt{candidates} register. 
    \item Repeat $\lceil n d\log(d/\epsilon)\rceil$ times:
    \begin{enumerate}
        \item Take the state from the \texttt{anti-symmetric} register and a randomly selected qudit from the \texttt{candidates} register, then apply~\Cref{alg:anti-sym-growing} to them. 
        \item If the run of \Cref{alg:anti-sym-growing} from (a) returns \textit{success}, increase $k$ by 1, and add the returned state to the \texttt{anti-symmetric} register. If $k = d$ or \texttt{candidates} is empty, go to step 4.
        \item If the run of \Cref{alg:anti-sym-growing} from (a) returns \textit{failure}, add the first $k$ qudits of the returned state to the \texttt{antisym\allowbreak metric} register and the last qudit of the returned state to the \texttt{candidates} register. 
    \end{enumerate}
    \item Return the state in the \texttt{anti-symmetric} register. 
\end{enumerate}
\end{longfbox}

\begin{theorem}
    \label{thm:sample-efficient-analysis}
    \Cref{alg:sample-efficient-antisym} is correct, that is, it always returns an anti-symmetric state. Furthermore if the input state is of the form $\rho^{\otimes n}$, where $\rho$ has rank $r$ and minimum non-zero eigenvalue $\lambda_{\mathrm{min}}$, then the algorithm returns the unique anti-symmetric state over the support of $\rho$ with probability at least $1 - \left(\epsilon+\exp(-\Omega(n\lambda_{\mathrm{min}}))\right)$ provided $n \geq C r/\lambda_{\text{min}}$ for a sufficiently large constant $C>0$.
    
    Explicitly, if 
    \begin{align}
        \rho = \sum_{i=1}^r \lambda_i \ketbra{\varphi_i}{\varphi_i}
    \end{align}
    then~\Cref{alg:sample-efficient-antisym} returns a state proportional to
    \begin{align}
        \Pi_{\mathrm{anti}}^r \left(\ket{\varphi_1}\otimes\ket{\varphi_2} \otimes ... \otimes \ket{\varphi_r}\right)
    \end{align}
    with probability as bounded above. 
\end{theorem}

\begin{proof}
    Correctness follows immediately from correctness of~\Cref{alg:anti-sym-growing}. The analysis of the state returned by the algorithm is most cleanly expressed using techniques from Schur-Weyl duality, and so is deferred to~\Cref{sec:schur_and_purification}. 
\end{proof}

\begin{remark}
    The proof of \Cref{thm:sample-efficient-analysis} also shows the following: If we are promised in advance that the unknown state $\rho$ has rank $r$, then it suffices to take $\lceil n r\log(r/\epsilon)\rceil$ repetitions in Step 3 of \Cref{alg:sample-efficient-antisym}.
\end{remark} \section{Connections to Schur-Weyl duality}
\label{sec:schur_and_purification}

In this section we  discuss the anti-symmetrization algorithm developed in~\Cref{sec:antisym_sec} from the perspective of Schur-Weyl duality. This perspective gives two new insights: an analysis of the sample-efficient anti-symmetrization algorithm (\Cref{alg:sample-efficient-antisym}) and an extension of that algorithm to give a new algorithm for optimal quantum purity amplification~\cite{li2024optimal}. The latter algorithm gives a qudit generalization of the swap-test-based optimal quantum purity amplification protocol for qubits developed in~\cite{brahmachari2025optimal}, answering an open question from that work. We begin this section with a review of some necessary concepts from Schur-Weyl duality, then present the analysis of~\Cref{alg:sample-efficient-antisym} and its extension. 

\subsection{Schur-Weyl background}

This subsection reviews some facts from the theory of Schur-Weyl duality. We review only what is required for the proofs in this paper, and so many foundational details of the theory are omitted. We refer the reader to~\cite[Chapter 2]{wright2016learn} for a more complete introduction to the topic and its relationship with quantum information. 

The key insight we require is that $\left(\mathbb{C}^d\right)^{\otimes n}$ is isomorphic to a direct sum of tensor products of irreducible modules of the symmetric group $S_n$ and the general linear group $\operatorname{GL}_d$. We can write this correspondence as
\begin{align}
\label{eq:Schur-Weyl}
\left(\mathbb{C}^d\right)^{\otimes n} \simeq \bigoplus_{\lambda \,\vdash_d \,n} V_\lambda \otimes W_\lambda\, ,
\end{align}
where each $W_\lambda$ is an irreducible module of $\operatorname{GL}_d$ and each $V_\lambda$ is an irreducible module of $S_n$.
It is worth noting that the isomorphism mapping from $\left(\mathbb{C}^d\right)^{\otimes n}$ to $\bigoplus_{\lambda \,\vdash_d \,n} V_\lambda \otimes W_\lambda$, which we commonly refer to as the full Schur transform, is in general not unique. In particular, to define a concrete Schur transform, one needs to specify a proper basis for both $W_\lambda$ and $V_\lambda$, e.g., the Gelfand-Tsetlin basis for $W_{\lambda}$ and the Young-Yamanouchi basis for $V_{\lambda}$. 
The particular choice of basis will not be important for our application. Nonetheless, we will arbitrarily fix a basis and denote the corresponding Schur transform unitary operator by $\Uschur: 
\left(\mathbb{C}^d\right)^{\otimes n} \mapsto \bigoplus_{\lambda \,\vdash_d \,n} V_\lambda \otimes W_\lambda
$ for ease of discussion.

The irreducible modules $V_\lambda$ and $W_\lambda$ are indexed by \textit{partitions} of $n$ into at most $d$ parts, denoted $\lambda \vdash_d n$. Each such $\lambda$ is a non-increasing sequence of at most $d$ positive integers that sum to $n$. We represent these partitions using Young diagrams: sequences of boxes with lengths of the rows corresponding to the integers in the partition. For example, 
\begin{align}
    \lambda = (4,2,1)= \ydiagram{4,2,1} 
\end{align}
is a partition of 7 into three parts. We use the notation $\lambda_i$ for the length of the $i$-th row of the Young diagram, and $\lambda'_i$ for the length of the $i$-th column.

For any given box (or cell) of a Young tableau, its content is defined to be its column index minus its row index. For example, the shaded box in the diagram 
\begin{align}
    \ydiagram{4,2,1}*[*(lightgray)]{3+1} 
\end{align}
has content $4 - 1= 3$. For a given box $\,\square\,$ we denote its content by $c(\square)$. The content sum of a Young diagram, denoted $c(\lambda)$ is defined to be the sum of the contents of all its boxes. For our running example: 
\begin{align}
    c(\lambda) = \sum_{\square \, \in\, \lambda}c(\square) = 0 -1-2+1+0+2+3 =3.
\end{align}

We say a vector $\ket{\psi} \in \left(\mathbb{C}^d\right)^{\otimes n}$ is in the $\lambda$-isotypic component of $\left(\mathbb{C}^d\right)^{\otimes n}$ if $\Uschur \ket{\psi}$ lies in a $V_\lambda \otimes W_\lambda$ module \footnote{As we have said, $\Uschur$ is not unique. However, a standard fact from representation theory is that the isotypic component itself is unique regardless of the choice of $\Uschur$.}. We also denote the projector onto this isotypic component by $\Pi_\lambda$.

For any $\sigma \in S_n$, the permutation operator $P_{\sigma}$ can be block diagonalized into an operator that acts trivially on the $W_{\lambda}$ modules under the Schur transform, i.e., 
\begin{align}
\Uschur P_\sigma \Uschur^\dagger = \bigoplus_{\lambda \,\vdash \, n} D_\lambda(\sigma)\otimes I_{W_\lambda} \, ,
\end{align} 
where $D_\lambda(\sigma)$ is the representation of $\sigma$ under the irrep $V_{\lambda}$, and $I_{W_\lambda}$ is the identity under the irrep $W_{\lambda}$.

This also means it is sometimes possible to characterize the action of some collection of permutations on a state just by knowing the isotypic component the state lies in. We give one example next. 

\begin{lemma}
    \label{lem:transposition_success}
    Let $\ket{\psi}$ be a state in the $\lambda$-isotypic component of $\left(\mathbb{C}^d\right)^{\otimes n}$ and define 
    \begin{align*}
        T([n]) = \sum_{i,j \in [n], \ i < j}P_{(i,j)}.
    \end{align*}
    Then we have 
    \begin{align}
        T([n]) \ket{\psi} = c(\lambda) \ket{\psi} \, .
    \end{align}
\end{lemma}

\begin{proof}
    This is, or at least is equivalent to, a standard result in the representation theory of the symmetric group~\cite{okounkov1996new,jucys1974symmetric,murphy1981new}. We briefly sketch the main ideas here. The starting point is to observe that the sum of transpositions $\sum_{i,j \in [n], i<j} (i,j)$ is central in the group algebra of the symmetric group, meaning it commutes with every other element in the algebra. Applying Schur's lemma then gives that for any vector $\ket{x}$ that lies in some irreducible representation of type $\lambda$ it holds 
    $T([n]) \ket{x} = c(\lambda) \ket{x}$ for some constant $c(\lambda)$.
    For a vector $\ket{\psi}$ within the $\lambda$-isotypic component, we can decompose $\ket{\psi}$ into its projections $\{ \ket{x_i} \}$ within each irreducible representations of type $\lambda$ within the isotypic component. We thus immediately have that $T([n]) \ket{\psi}= \sum_i T([n]) \ket{x_i} = c(\lambda) \sum_i \ket{x_i}  = c(\lambda) \ket{\psi}$.
    This concludes the proof of \Cref{lem:transposition_success}.
    \end{proof}

    Another important concept in our work is that of \textit{restrictions} of isotypic components. Consider an $n$-qudit state $\ket{\psi}$ in a $\lambda$-isotypic component. Assume further that a subset of $k$ qudits of that state are in a $\mu$-isotypic component, meaning that if we project those qudits onto the $\mu$-isotypic component, it leaves the state $\ket{\psi}$ unchanged (or, equivalently, if we trace out all but those $k$ qudits, the resulting state is supported on the $\mu$-isotypic component). The Littlewood-Richardson rule then describes the possible isotypic components which the complementary subset of $n-k$ qudits of $\ket{\psi}$ can be supported on. We will be interested in a special case of this rule, called Pieri's rule, which considers the same situation when the $k$ qudits are supported on the $\mu = (1,1,1,..,1)$-isotypic subspace. We state this next. 

    \begin{lemma}[Pieri's Rule, e.g., see Theorem 2.7.1 of \cite{grinberg2020hopf}]
    \label{lem:Pieri}
    Let $\ket{\psi}$ be an $n$-qudit state in a $\lambda$-isotypic component. Assume that a subset of $k$ qudits of $\ket{\psi}$ are in the $\mu = (1,1,...,1)$-isotypic component. Then the $n-k$ complementary qubits of $\ket{\psi}$ can only have support overlapping $\nu$-isotypic components of partitions $\nu$ that can be obtained by removing exactly one box from $k$ distinct rows of $\lambda$.
    \end{lemma}

    As an illustrative example, assume a 6-qudit state $\ket{\psi}$ is in the $\lambda$-isotypic component with
    \begin{align}
        \lambda = (2,2,2) = \ydiagram{2,2,2}\, .
    \end{align}
    If the first two qudits of $\ket{\psi}$ are in the $(1,1)$-isotypic component, then the mixed state $\text{Tr}_{12}\left[\ketbra{\psi}{\psi}\right]$ may be supported on the $(2,1,1)$-isotypic component, since it can be obtained by deleting one box from the second row of $\lambda$ and one from the third, as indicated:
    \begin{align}
        \ydiagram{2,2,2}*[*(lightgray)]{0,1+1,1+1}  \rightarrow_{\text{allowed deletion}} \ydiagram{2,1,1}\,.
    \end{align}
    However, this mixed state may not be supported on the $(2,2)$-isotypic component (despite it having the correct number of overall boxes) because this transition would require deleting two boxes from the same row:
        \begin{align}
        \ydiagram{2,2,2}*[*(lightgray)]{0,0,2}  \rightarrow_{\text{forbidden deletion}} \ydiagram{2,2}\,.
    \end{align}

    There is a close relation between Pieri's rule and the anti-symmetric subspace, coming from the following fact. 

    \begin{lemma}
        \label{lem:anti_isotypic}
        For any state $\ket{\psi}$ on $n$ qudits and for any subset of qudits $S$, the following are equivalent:
        \begin{enumerate}[label = (\roman*)]
            \item $\Pi^S_{\mathrm{anti}} \ket{\psi} = \ket{\psi}$. 
            \item The qudits of $\ket{\psi}$ in $S$ are in the $(1,1,...,1)$-isotypic component. 
        \end{enumerate}        
    \end{lemma}
    \begin{proof}
        This follows quickly from some standard results in representation theory. For completeness, we sketch the argument here, although some background is omitted. 
        In particular, the isotypic subspace of type $(1, 1, \cdots, 1)$ is of form $ \Span \{  \; U^{\otimes n} P_\pi  \Pi^S_{\mathrm{anti}} \ket{w} \forall  \pi \in S_n, U \in GL_d\}$, where $\ket{w} \in (\C^d)^{\otimes n}$ is some state that makes the subspace non-vanishing.
        It is then not hard to see that this is exactly the anti-symmetric subspace.
    \end{proof}

    \subsection{\texorpdfstring{Analysis of~\Cref{alg:sample-efficient-antisym}}{Analysis of sample-efficient anti-symmetrization}}

    We now use the machinery developed in the previous section to analyze the sample-efficient anti-symmetrization algorithm given in~\Cref{alg:sample-efficient-antisym}. The key tool used in this analysis is the following lemma, which we will use to lower bound the probability of a measurement successfully growing the size of the \texttt{anti-symmetric} register during step 3 of the algorithm.

    \begin{lemma}
    \label{lem:b_growth}
        Let $\ket{\psi}$ be an $n$-qudit state in the $\lambda$-isotypic component and recall that $\lambda_1'$ is the length of the first column of $\lambda$. Assume there is a subset of qudits $S \subseteq [n]$ with $|S| = k$ and 
        \begin{align}
            \Pi^S_{\mathrm{anti}} \ket{\psi} = \ket{\psi}. 
        \end{align}
        Then define the operator 
        \begin{align}
            B^S = \frac{1}{n-k} \left(\sum_{j \in [n]\backslash S} A^{S,j}\right).
        \end{align}
        We have 
        \begin{align}
            \bra{\psi}B^S \ket{\psi} = 0 
        \end{align}
        whenever $k = \lambda_1'$ and 
        \begin{align}
            \bra{\psi}B^S \ket{\psi} \geq \frac{\lambda_1' - k}{(n-k)(k+1)}
        \end{align}
        otherwise. 
    \end{lemma}
    \begin{proof}
        We first extend the notation of~\Cref{lem:transposition_success} and, for any subset $R \subseteq [n]$ define 
        \begin{align}
            T(R) = \sum_{i, j \in R, \ i < j } P_{\transposition{i}{j}}.
        \end{align}
        Then we rewrite the operator $B^S$ as 
        \begin{align}
            B^S &= \frac{1}{n-k} \left(\sum_{j \in [n]\backslash S} A^{S,j}\right) \\
            &= \frac{1}{(n-k)(k+1)}\left(\sum_{j \in [n]\backslash S}\left( I - \sum_{i \in S}P_{\transposition{i}{j}} \right)\right) \\
            &=\left(\frac{1}{k+1}\right) I - \frac{1}{(n-k)(k+1)}\sum_{ i \in S} \sum_{ j \in [n] \backslash S} P_{(i,j)}\\
            &= \left(\frac{1}{k+1}\right) I - \frac{1}{(n-k)(k+1)}\left(T([n]) - T(S) - T([n]\backslash S) \right)\, .
            \label{eq:b_lem_op_decomp}
        \end{align}
        Next, by~\Cref{lem:Pieri,lem:anti_isotypic} we can write 
        \begin{align}
            \ket{\psi} = \sum_{\nu \,\in \,\mathcal{V}} \alpha _\nu \ket{\psi_\nu}
            \label{eq:b_lem_state_decomp}
        \end{align}
        where we have 
        \begin{align}
            \left(\Pi^{S}_{anti} \otimes \Pi^{S^c}_\nu\right) \ket{\psi_\nu} = \ket{\psi_\nu}
        \end{align}
        where $\Pi^{S^c}_\nu$ denotes the projector onto the $\nu$ isotypical component applied to the qudits in $S^c$ and the sum runs over a set $\mathcal{V}$ of all partitions which can be obtained from $\lambda$ by deleting exactly one box from $k$ different rows. We also fix $\ket{\psi_\nu}$ to be properly normalized, so $\sum_\nu \alpha_\nu^2 = 1$. Combining~\Cref{eq:b_lem_state_decomp,eq:b_lem_op_decomp} and using~\Cref{lem:transposition_success} gives
        \begin{align}
            \bra{\psi} B^S \ket{\psi} &= \frac{1}{k+1}- \frac{1}{(n-k)(k+1)}\left(c(\lambda) -c(1,1,...,1) - \sum_{\nu \in \mathcal{V}} \alpha_\nu^2 c(\nu) \right) \\
            &= \frac{1}{k+1}- \frac{1}{(n-k)(k+1)}\left(c(\lambda) + \frac{k(k-1)}{2} - \sum_{\nu \in \mathcal{V}} \alpha_\nu^2 c(\nu) \right) \\
            &= \frac{1}{k+1} -  \frac{1}{(n-k)(k+1)}\sum_{\nu \in \mathcal{V}} \alpha_\nu^2 \left(\frac{k(k-1)}{2} +  c(\lambda) - c(\nu) \right)\, ,\label{eq:b_lem_pre_simplified}
        \end{align}
        where explicitly computed the content of the $(1,1,...,1)$ partition on the second line above, and rearranged terms on the third. 
        
        Next, we analyze the expression above when $k = \lambda_1'$. In this case $\mathcal{V}$ contains a single irrep $\nu$ obtained from $\lambda$ by deleting one box from each row of $\lambda$. The difference $c(\lambda) - c(\nu)$ is equal to the content of each of the deleted boxes, so we have
        \begin{align}
            c(\lambda)- c(\nu) =\sum_{i \in [\lambda_1']} (\lambda_i - i) =  n - \frac{k(k+1)}{2} \, ,
        \end{align}
        and then~\Cref{eq:b_lem_pre_simplified} simplifies to
        \begin{align}
            \bra{\psi} B^S \ket{\psi} &= \frac{1}{k+1} -  \frac{1}{(n-k)(k+1)}\left(\frac{k(k-1)}{2} + n - \frac{k(k+1)}{2}  \right) \\
            &= \frac{1}{k+1} -\frac{1}{k+1} = 0.
        \end{align}

        Now, if $k < \lambda_1'$, then $\mathcal{V}$ contains all irreps obtained by deleting boxes from $k$ different rows of $\lambda$. Let $\nu$ be one such irrep, and let $R \subset [\lambda_1']$ label the rows from which boxes are deleted. Then 
        \begin{align}
            c(\lambda)- c(\nu) &=\sum_{i \in R} (\lambda_i - i) \\
            &\leq \left(\sum_{i \in R} \lambda_i \right)- \frac{k(k+1)}{2} \\
            &\leq (n - (\lambda_1' - k))- \frac{k(k+1)}{2} \, ,
        \end{align}
        since each of the $\lambda_1' - k$ rows not in $R$ contains at least one box. Plugging this in to~\Cref{eq:b_lem_pre_simplified} for each $\nu \in \mathcal{V}$ gives
        \begin{align}
            \bra{\psi} B^S \ket{\psi} &\geq \frac{1}{k+1} -  \frac{1}{(n-k)(k+1)}   \left(\frac{k(k-1)}{2} +n - \frac{k(k+1)}{2} - (\lambda_1' - k) \right) \\
            &= \frac{\lambda_1' - k}{(n-k)(k+1)}\,,
        \end{align}
        as desired. 
    \end{proof}

    To relate~\Cref{lem:b_growth} and~\Cref{alg:sample-efficient-antisym} we will view the input $\rho^{\otimes n}$ to the algorithm as a probabilistic mixture of pure states, and then analyze the evolution of these pure states under~\Cref{alg:sample-efficient-antisym} depending on the $\lambda$-isotypic component they are contained in. The last step is a simple corollary of~\Cref{lem:b_growth}, stated next. 
    
    \begin{corollary}
        \label{cor:sample_efficient_left_column}
        Let $\ket{\psi}$ be a state in the $\lambda$-isotypic component of $(\mathbb{C}^d)^{\otimes n}$. If~\Cref{alg:sample-efficient-antisym} is applied to $\ket{\psi}$, it returns an anti-symmetric state on $k = \lambda_1'$ qudits with probability at least $1 - \epsilon$.
    \end{corollary}

    Note that, by convexity, if a mixed state $\rho$ is supported on the $\lambda$-isotypic component, then the same conclusion as in \Cref{cor:sample_efficient_left_column} holds.

    \begin{proof}
        Let $\ket{\tilde{\psi}}$ describe the state stored in the \texttt{anti-symmetric} and \texttt{candidate} registers at some arbitrary intermediate stage of the algorithm, and let $S \subseteq [n]$ denote the qudits of $\ket{\tilde{\psi}}$ in the \texttt{anti-symmetric} register.~\Cref{alg:anti-sym-growing} only acts on $\ket{\psi}$ by representations of the symmetric group algebra, which are block diagonal with respect to the isotypic components, so $\ket{\tilde{\psi}}$ is also in the $\lambda$-isotypic component. By construction,  step 3(a) of~\Cref{alg:sample-efficient-antisym} returns success with probability
        \begin{align}
            \bra{\tilde{\psi}} \frac{1}{n-k}\sum_{j \in [n] \backslash S} A^{S,j} \ket{\tilde{\psi}} = \bra{\tilde{\psi}} B^{S} \ket{\tilde{\psi}}.
        \end{align}
        But then~\Cref{lem:b_growth} gives that each iteration of step 3 grows the size of the \texttt{anti-symmetric} register from $k$ to $k+1$ with probability 
        \begin{align}
            \frac{\lambda_1' - k}{(n-k)(k+1)} \geq \frac{\lambda_1' - k}{n \lambda_1'}.
        \end{align}
        Relating this to the probability of a coupon collector with $\lambda'_1$ coupons and a probability $(n\lambda'_1)^{-1}$ of picking up any coupon (as in the proof of~\Cref{thm:sample-inefficient-analysis}) shows that after $t$ repetitions, the probability that the $\texttt{anti-symmetric}$ register does not contain $\lambda'_1$ qudits is upper bounded by 
        \begin{align}
            \lambda_1'\left(1-\frac{1}{n\lambda_1'}\right)^t\leq \lambda'_1e^{-t/n\lambda_1'}.
        \end{align}
        Setting $t = \lceil n d(\log(d) + \log(1/\epsilon))\rceil$ and noting $d \geq \lambda_1'$ proves the result. 
    \end{proof}

    We now have all the tools required to complete the proof of~\Cref{thm:sample-efficient-analysis} and the analysis of~\Cref{alg:sample-efficient-antisym}. We present this next. 

    \begin{proof}[Proof (\Cref{thm:sample-efficient-analysis}, continued).] First note that~\Cref{alg:sample-efficient-antisym} acts on the input state only through representations of the symmetric group algebra and so never mixes isotypic components. \Cref{cor:sample_efficient_left_column} and a union bound then give: when ~\Cref{alg:sample-efficient-antisym} is applied to any pure state that has overlap $1-\delta$ with $\lambda$-isotypic components satisfying $\lambda_1'=r$, it returns an anti-symmetric state on $r$ qudits with probability at least $1 -\delta - \epsilon$. This observation extends immediately to mixed states. But now standard arguments (for example, Remark 2.25 of~\cite{o2021quantum} plus a geometric sum argument and concentration inequality)
    give that if a mixed state $\rho$ has rank $r$ and minimum non-zero eigenvalue $\lambda_{\text{min}}$, sampling from the Schur-Weyl distribution of $\rho^{\otimes n}$ yields a partition $\lambda$ for which $\lambda_1'=r$ with probability $1 - \exp(-\Omega( n \lambda_{\text{min}}))$ provided $n \geq Cr/\lambda_{\text{min}}$ for a sufficiently large constant $C>0$. 
    The result follows. 
    \end{proof}

    \subsection{Algorithms for Schur sampling and purity amplification}\label{subsec:schur-sampling-purification}

    We begin with some observations about~\Cref{alg:sample-efficient-antisym}. For any partition $\lambda \vdash n$, \Cref{cor:sample_efficient_left_column} above shows that if this algorithm is run on a pure state $\ket{\psi}$ in the $\lambda$-isotypic component of $\left(\mathbb{C}^d\right)^{\otimes n}$, it returns an anti-symmetric state on $\lambda_1'$ qudits with probability $1-\epsilon$. In this case, Pieri (\Cref{lem:Pieri}) then tells us the system in the \texttt{candidates} register must be in the $\nu$-isotypic component, where $\nu \vdash (n-\lambda_1')$ is the unique partition obtained by deleting one box from each row of $\lambda$. If we then run~\Cref{alg:sample-efficient-antisym} again on the \texttt{candidates} register, it will output an anti-symmetric state on $\nu_1' = \lambda_2'$ qudits. Continuing in this way we can learn the lengths $\lambda_1', \lambda_2', ... \lambda_{\lambda_1}'$ of each column in $\lambda$. 
    
    If this same algorithm is run on a state $\ket{\psi}$ that is in a superposition of different isotypic components, the algorithm will output some $\lambda_1', ...,\lambda_{\lambda_1}'$ with probability corresponding to the overlap of $\ket{\psi}$ with the corresponding $\lambda$-isotypic component (and some error depending on $\epsilon$).
    This is exactly the task of \textit{weak Schur sampling}. We present this algorithm next.

\begin{longfbox}[breakable=false, padding=1em, margin-top=1em, margin-bottom=1em]
\begin{algorithm}
Weak Schur Sampling
\label{alg:weak-Schur-Sampling}
\end{algorithm}
\noindent \textbf{Input:} An unknown state $\rho \in \mathcal{L}\left(\left(\mathbb{C}^d\right)^{\otimes n}\right)$ and an accuracy parameter $\delta$. \\
\noindent \textbf{Output:} A partition $\lambda$ drawn from a probability distribution $p_\lambda$ satisfying
\begin{align}
    \sum_{\lambda \, \vdash \,n} \left|p_\lambda - \text{Tr}[\Pi_\lambda \rho]\right| \leq \delta\, .
\end{align}
\begin{enumerate}
    \item Label a register \texttt{unsampled} and add $\rho$ to it. Initialize a parameter $j = 1$.
    \item While \texttt{unsampled} is not empty:
    \begin{enumerate}
        \item Take the state from the \texttt{unsampled} register and apply~\Cref{alg:sample-efficient-antisym} to it with parameter $\epsilon = \delta/2n$. Add the state in the \texttt{candidates} register of this algorithm to the \texttt{unsampled} register. 
        \item Set $\lambda_j'$ equal to the length of the anti-symmetric state returned by~\Cref{alg:sample-efficient-antisym}. Increment $j$.
    \end{enumerate}
    \item Return the partition $\lambda$ with column lengths $\lambda_1', \lambda_2', ..., \lambda_{j-1}'$.
\end{enumerate}
\end{longfbox}

\begin{theorem}
    \label{thm:correctness_weak_schur}
    \Cref{alg:weak-Schur-Sampling} is correct, that is, it returns a partition $\lambda$ drawn from a probability distribution of the promised form. 
    It requires at most $2n\lceil nd\log(2nd/\delta)\rceil$ controlled-SWAP operations and has $T$-count $O\left(n^2d\log(d)\log^2(nd/\delta)\right)$. 
\end{theorem}

\begin{proof}

    First consider the case where the algorithm is run on a pure state $\rho = \ketbra{\psi}{\psi}$ in the $\lambda$-isotypic component of $\left(\mathbb{C}^d\right)^{\otimes n}$ for some arbitrary $\lambda \vdash n$. Let $\rho_i$ be the state stored in the \texttt{unsampled} register after step 2 of the algorithm has repeated $i$ times (so $\rho_0 = \rho$). For every positive integer $i \leq \lambda_1$ let $\nu(i)$ be the partition obtained by deleting the $i$ leftmost columns of $\lambda$. Assume for induction that $\rho_i$ is in the $\nu(i)$-isotypic component. Then Pieri (\Cref{lem:Pieri}) gives that $\rho_{i+1}$ is in the $\nu(i+1)$-isotypic component provided the $(i+1)$-th repetition of step 2 of~\Cref{alg:weak-Schur-Sampling} returns an anti-symmetric state on ${\nu(i)}'_1$ qudits. \Cref{cor:sample_efficient_left_column} gives that this happens with probability $1-\epsilon = 1-\delta /2n$. Then a union bound argument gives that~\Cref{alg:weak-Schur-Sampling} returns the correct partition $\lambda$ with probability at least
    $
        1 - \lambda_1(\delta/2n) \geq 1- \delta/2. 
    $

    Now when the input is an arbitrary mixed state $\rho$, as \Cref{alg:weak-Schur-Sampling} does not mix isotypic components, the analysis above gives that with probability at least $1 - \delta/2$ the algorithm returns a partition $\lambda$ drawn exactly from the distribution 
    \begin{align}
        q_\lambda = \Tr[\Pi_\lambda \rho].
    \end{align}
    Letting $r_\lambda$ be the distribution output by the algorithm a $\delta/2$-fraction of the time which includes any time the algorithm does not output from the $q_\lambda$ distribution, we have
    \begin{align}
        p_\lambda = (1-\delta/2)q_\lambda + (\delta/2) r_\lambda \, ,
    \end{align}
    and so we see that, for any $r_\lambda$, we have
    \begin{align}
        \lVert p_\lambda-q_\lambda \rVert _1 \leq \delta,
    \end{align}
    as desired. 

    To bound the number of controlled-SWAP operations we note each application of~\Cref{alg:sample-efficient-antisym} to build an antisymmetric state on $\lambda_i'$ qudits requires at most $ \lceil nd\log(d/\epsilon)\rceil = \lceil nd\log(2nd/\delta)\rceil$ applications of~\Cref{alg:anti-sym-growing}, each of which requires at most $2\lambda_i'$ controlled-SWAP gates. So the total number of controlled-SWAP operations is upper bounded by
    \begin{align}
        \lceil nd\log(2nd/\delta)\rceil \sum_{i = 1}^{\lambda_1} 2\lambda_{i}' = 2n\lceil nd\log(2nd/\delta)\rceil
    \end{align}
    To convert to a $T$-count note each controlled-SWAP operation requires $O(\log(d))$ $T$ gates. The each application of~\Cref{alg:anti-sym-growing} also requires $\mathcal{O}(\log(k)) \leq \mathcal{O}(\log(d))$ single qubit rotations, each of which which can be implemented up to error $\eta$ using $\mathcal{O}(\log(\eta^{-1}))$ $T$ gates. Setting 
    \begin{align}
        \eta^{-1} = \lceil nd\log(2nd/\delta)\rceil/\delta
    \end{align}
    to keep total error bounded gives the result.

\end{proof}

In \Cref{alg:weak-Schur-Sampling} all quantum registers are discarded and the only output is a classical description of the partition $\lambda$. If the algorithm is altered so that quantum information is retained and included in the output, this modified algorithm performs the task of \textit{Unitary Schur Sampling}. The name refers to the fact that the algorithm destroys any superposition the input state has in each irreducible module of $S_n$ (by projecting onto a fixed vector within that module) but leaves superpositions within the irreducible modules of $\operatorname{GL}_d$ untouched. Informally: the algorithm destroys permutation information about the input state while preserving unitary information. 

\begin{longfbox}[breakable=true, padding=1em, margin-top=1em, margin-bottom=1em]
\begin{algorithm}
Unitary Schur Sampling
\label{alg:unitary-Schur-Sampling}
\end{algorithm}
\noindent \textbf{Input:} A state $\rho^{\otimes n}$ for unknown $\rho \in \mathcal{L}\left(\mathbb{C}^d\right)^{\otimes n}$ and an accuracy parameter $\delta$.  \\
\noindent \textbf{Output:} A state $\sigma$ satisfying 
\begin{align}
    \lVert \sigma - \sum_{\lambda \,\vdash \, n} \rho^{[\lambda]}\otimes \ketbra{\lambda}{\lambda} \rVert_1\leq \delta
\end{align}
where $\ket{\lambda}$ is a classical description of the partition $\lambda$ and $\rho^{[\lambda]}$ is the unnormalized mixed state defined  (in terms of the $V_\lambda$ and $W_\lambda$ modules, see~\Cref{eq:Schur-Weyl}) by
\begin{align}
    \Uschur \rho^{[\lambda]} \Uschur^\dagger = \ketbra{0_\lambda}{0_\lambda} \otimes \text{Tr}_{V_\lambda}[ \Uschur \Pi_\lambda \rho^{\otimes n} \Pi_\lambda \Uschur^\dagger ] \, ,
\end{align}
where $\ket{0_\lambda}$ is some reference state \footnote{The exact value of $\ket{0_\lambda}$ depends also on the choice of the basis of $\Uschur$, which is not important for our application.} in the $V_\lambda$ module fixed by the algorithm.
\begin{enumerate}
    \item Label registers \texttt{unsampled} and \texttt{sampled}. Add the state $\rho^{\otimes n}$ to the \texttt{unsampled} register and initialize a parameter $j = 1$.
    \item While \texttt{unsampled} is not empty:
    \begin{enumerate}
        \item Take the state in the \texttt{unsampled} register and apply~\Cref{alg:sample-efficient-antisym} to it with parameter $\epsilon = \delta/2n$. Add the state in the \texttt{antisymmetric} register of this algorithm to the \texttt{sampled} register and the state in the \texttt{candidates} register of this algorithm to the \texttt{unsampled} register. 
        \item Set $\lambda_j'$ equal to the length of the anti-symmetric state returned by~\Cref{alg:sample-efficient-antisym}. Increment $j$.
    \end{enumerate}
    \item Return the state in the \texttt{sampled} register along with $\ket{\lambda}$ where $\lambda$ is the partition with column lengths $\lambda_1', \lambda_2', ..., \lambda_{j-1}'$.
\end{enumerate}
\end{longfbox}

\begin{theorem}
    \label{thm:unitary-Schur-Sampling}
    \Cref{alg:unitary-Schur-Sampling} is correct, that is, it returns a state $\sigma$ with the promised form. It requires at most $2n\lceil nd\log(2nd/\delta)\rceil$ controlled-SWAP operations and has $T$-count $O\left(n^2d\log(d)\log^2(nd/\delta)\right)$.
\end{theorem}

\begin{proof}
    Because our input is block diagonal over isotypical components, i.e. 
    \begin{align}
        \sum_\lambda \Pi_\lambda\rho^{\otimes n}\Pi_\lambda = \rho^{\otimes n}
    \end{align}
    it suffices to analyze the algorithm over each isotypical component. For one such $\lambda$, let  $G_\lambda$ be the event that each application of~\Cref{alg:sample-efficient-antisym} works as intended, i.e. that the $i$-th call to the algorithm returns an antisymmetric state on $\lambda_i'$ qudits for all $i \in [\lambda_1]$. By the same analysis as in~\Cref{thm:correctness_weak_schur} we have 
    \begin{align}
        \text{Pr}[G_\lambda] \geq 1- \delta/2
    \end{align} for all $\lambda$. 
    
    Let $\sigma_\lambda$ be the (normalized) state in the \texttt{sampled} register conditioned on event $G_\lambda$ when the algorithm returns partition label $\lambda$. The same arguments as in ~\Cref{thm:correctness_weak_schur} shows that $\sigma_\lambda$ is in the $\lambda$-isotypic component with probability at least $1 - \delta/2$, that is, 
    \begin{align}
        \text{Tr}[\Pi_\lambda\sigma_\lambda] \geq 1-\delta/2.
    \end{align}
    Correctness of~\Cref{alg:sample-efficient-antisym} also gives that 
    \begin{align}
        \text{Tr}\left[\left(\Pi^{\lambda_1'}_{\mathrm{anti}} \otimes \Pi^{\lambda_2'}_{\mathrm{anti}} \otimes ... \otimes \Pi^{\lambda_{\lambda_1}'}_{\mathrm{anti}}\right) \sigma_\lambda\right] = 1.
    \end{align}
    For any $S \subseteq [n]$ let 
    \begin{align}
        A(S) = \frac{1}{|S|!}\sum_{\substack{\sigma \in S_n  \\  \sigma(i)  = i\, \forall \,i \notin S}} \text{sign}(\sigma)\sigma
    \end{align}
    be the element of the symmetric group algebra corresponding to an antisymmetrizer on the elements of $S$ and let 
    \begin{align}
        A(\lambda_1',\lambda_2',...,\lambda_{\lambda_1}') = \prod_{i =1}^{\lambda_1} A\left(\left\{\sum_{j=1}^{i-1} \lambda_j' + 1, \sum_{j=1}^{i-1} \lambda_j'+2, ...,\sum_{j=1}^i \lambda_j' \right\}\right)
    \end{align}
    so $\Pi^{\lambda_1'}_{\mathrm{anti}} \otimes \Pi^{\lambda_2'}_{\mathrm{anti}} \otimes ... \otimes \Pi^{\lambda_{\lambda_1}'}_{\mathrm{anti}}$ is exactly the natural representation of $A(\lambda_1',\lambda_2',...,\lambda_{\lambda_1}')$. 
    Standard representation theory techniques (see~\Cref{lem:one_dimensional} in~\Cref{app:Additional_rep_theory_proofs} for a proof) give that there is a unique state $\ket{0_\lambda} \in V_\lambda$ satisfying 
    \begin{align}
        D_\lambda \left(A(\lambda_1',\lambda_2',...,\lambda_{\lambda_1}')\right) \ket{0_\lambda} = \ket{0_\lambda} \, .
    \end{align}
    For example, this state is exactly the column ordered Young-Yamanouchi vector of~\cite{li2024optimal}. This shows that $\sigma_\lambda$ must overlap the $V_\lambda$ module on exactly the state $\ketbra{0_\lambda}{0_\lambda}$, that is, we must have 
    \begin{align}
        \Pi_\lambda \sigma_\lambda \Pi_\lambda = \ketbra{0_\lambda}{0_\lambda} \otimes \sigma_\lambda^W
    \end{align}
    for some state $\sigma_\lambda^W \in \mathcal{L}(W_\lambda)$. 
    \Cref{alg:sample-efficient-antisym} only ever acts on the state $\rho$ with representations of the symmetric group algebra, which act as the identity on the $W$ modules, so we also have
    \begin{align}
        \sigma_\lambda^W &\propto \sum_r\text{Tr}_{V_\lambda}[(K_r\otimes I)
        \Uschur \Pi_\lambda \rho^{\otimes n} \Pi_\lambda \Uschur^\dagger (K_r\otimes I)^\dagger] \\
        &= \text{Tr}_{V_\lambda}[
        \Pi_\lambda \rho^{\otimes n}  \Pi_\lambda].
    \end{align}
    where $\{K_r\}$ are Kraus operators describing the action of the algorithm on the $V_\lambda$ module.
    We conclude 
    \begin{align}
        \sigma &= \left(1-\frac{\delta}{2}\right) \left(\sum_{\lambda \,\vdash \, n} \ketbra{0_\lambda}{0_\lambda} \otimes \text{Tr}_{V_\lambda}[
        \Pi_\lambda \rho^{\otimes n} \Pi_\lambda] \otimes \ketbra{\lambda}{\lambda} \right) + \frac{\delta}{2} (\sigma') \\
        &= \left(1 - \frac{\delta}{2}\right)\left(\sum_{\lambda \,\vdash \, n} \rho^{[\lambda]}\otimes \ketbra{\lambda}{\lambda}\right) + \frac{\delta}{2} (\sigma')
    \end{align}
    for some (normalized) state $\sigma'$ and correctness of the algorithm follows.

    The required number of controlled-SWAPs and T-count of the algorithm follows immediately from the analysis of~\Cref{alg:weak-Schur-Sampling}.
\end{proof}

We end this section with a discussion of \textit{quantum purity amplification}. Given an input $\rho^{\otimes n} \in \left(\mathcal{L}(\mathbb{C}^d)^{\otimes n}\right)$, the goal of purity amplification (in its simplest form) is to output a state $\rho' \in \mathcal{L}(\mathbb{C}^d)$ that is concentrated on the top eigenvector of $\rho$, that is, to return a state maximizing
\begin{align}
    \bra{\varphi_1} \rho'\ket{\varphi_1} \, ,
\end{align}
where $\ket{\varphi_1}$ is the largest eigenvector of $\rho$. This task was formalized in~\cite{li2024optimal}, which also gave an optimal protocol for it based on measuring the information in the $V_\lambda$ modules of the state $\rho^{\otimes n}$ and then applying a measurement-dependent correction to the resulting state. The follow-up work~\cite{brahmachari2025optimal} gave an alternate protocol for optimal quantum purity amplification of qubit input states based on random SWAP tests. We show that a simple modification of~\Cref{alg:unitary-Schur-Sampling} extends the ideas of~\cite{brahmachari2025optimal} to give an optimal protocol for quantum purity amplification of arbitrary qudit input states. Before presenting the formal algorithm, we give an intuitive explanation of why preparing antisymmetric states can be a helpful subroutine for purity amplification. 

Consider~\Cref{alg:unitary-Schur-Sampling} acting on an input state $\rho^{\otimes n}$. Label the eigenvectors of $\rho$ in the order of non-increasing eigenvalue by $\ket{\varphi_1}, \ket{\varphi_2},...,\ket{\varphi_r}$. Any $k$-qudit antisymmetric state produced by the algorithm must overlap equally with $k$ different eigenvectors of $\rho$. Preparing such a state and removing it from the system then has the effect of ``pumping out'' asymmetry from the system, leaving the remaining system more concentrated on the most common eigenvector $\ket{\varphi_1}$. Following this intuition suggests that the final antisymmetric state prepared by~\Cref{alg:unitary-Schur-Sampling} will be the one most concentrated on the top eigenvector of $\rho$, and returning a random qudit from this state (or the state itself, in the case that it is a one-qudit state already) will provide a solution to the task of purity amplification. We now formalize this process, and show (by reducing to the analysis of~\cite{li2024optimal}) that this intuition is correct. 

\begin{longfbox}[breakable=true, padding=1em, margin-top=1em, margin-bottom=1em]
\begin{algorithm}
Quantum Purity Amplification
\label{alg:quantum_purity_amplification}
\end{algorithm}
\noindent \textbf{Input:} An state $\rho^{\otimes n}$  for unknown $\rho \in \mathcal{L}\left(\mathbb{C}^d\right)$ and an accuracy parameter $\delta$.  \\
\noindent \textbf{Output:} A state $\rho'$ satisfying 
\begin{align}
    \lVert \rho' - \rho_{\text{opt}} \rVert_1\leq \delta \, ,
\end{align}
where $\rho_{\text{opt}}$ is the state returned by the optimal quantum purity amplification protocol of~\cite{li2024optimal}.
\begin{enumerate}
    \item Label registers \texttt{unsampled} and \texttt{purified}. Add the state $\rho^{\otimes n}$ to the \texttt{unsampled} register.
    \item While \texttt{unsampled} is not empty:
    \begin{enumerate}
        \item Remove the state in the \texttt{unsampled} register and apply~\Cref{alg:sample-efficient-antisym} to it with parameter $\epsilon = \delta/2n$. \textbf{Replace} the state in the \texttt{purified} register with the state in the \texttt{antisymmetric} register of this algorithm, and add the state in the \texttt{candidates} register of this algorithm to the \texttt{unsampled} register. 
    \end{enumerate}
    \item Return the last qudit from the \texttt{purified} register. 
\end{enumerate}
\end{longfbox}

\begin{theorem}
    \label{thm:quantum_purity_amplification}
    \Cref{alg:quantum_purity_amplification} is correct, that is, it returns a state $\rho'$ with the desired properties. It requires at most $2n\lceil nd\log(2nd/\delta)\rceil$ controlled-SWAP operations and has $T$-count $O\left(n^2d\log(d)\log^2(nd/\delta)\right)$.
\end{theorem}

\begin{proof} 
We briefly adopt the notation of~\cite{li2024optimal}  in order to compare. The operational description of their protocol shows that it is equivalent to the following three steps: (i) measure the Young Diagram ($\lambda$) of the input along with the Young-Yamanouchi label ($s$) of the permutation register; (ii) apply a transition operator $\Pi^{[\lambda]}_{\lambda^\diamond,\, s}$ to map this state to the column ordered subspace ($\lambda^\diamond$);  (iii) return the register corresponding to the bottom box of the rightmost column of $\lambda^\diamond$, tracing out the rest.
In our notation, the state prepared after steps (i) and (ii) of this protocol is exactly 
\begin{align}
    \sum_{\lambda \, \vdash \,n} \text{Tr}[\Pi_\lambda\rho^{\otimes n}] \frac{\left(\Pi^{\lambda_1'}_{\mathrm{anti}} \otimes \Pi^{\lambda_2'}_{\mathrm{anti}} \otimes ... \otimes \Pi^{\lambda_{\lambda_1}'}_{\mathrm{anti}}\right)  \
    \Pi_{\lambda} \rho^{\otimes n} \Pi_{\lambda}\left(\Pi^{\lambda_1'}_{\mathrm{anti}} \otimes \Pi^{\lambda_2'}_{\mathrm{anti}} \otimes ... \otimes \Pi^{\lambda_{\lambda_1}'}_{\mathrm{anti}}\right)}{\text{Tr}\left[\left(\Pi^{\lambda_1'}_{\mathrm{anti}} \otimes \Pi^{\lambda_2'}_{\mathrm{anti}} \otimes ... \otimes \Pi^{\lambda_{\lambda_1}'}_{\mathrm{anti}}\right)  \Pi_{\lambda} \rho^{\otimes n} \Pi_{\lambda}\left(\Pi^{\lambda_1'}_{\mathrm{anti}} \otimes \Pi^{\lambda_2'}_{\mathrm{anti}} \otimes ... \otimes \Pi^{\lambda_{\lambda_1}'}_{\mathrm{anti}}\right)\right]} \otimes \ketbra{\lambda}{\lambda} \, .
\end{align}
(See~\Cref{app:Additional_rep_theory_proofs}, and in particular~\Cref{lem:YY_column_anti} for a justification of this fact.)
But the analysis of~\Cref{alg:unitary-Schur-Sampling} shows that it returns a state at most $\delta$ far in trace distance from the state above. \Cref{alg:quantum_purity_amplification} returns the last qudit of the state prepared by~\Cref{alg:unitary-Schur-Sampling}, and since trace distance is non-increasing under partial traces, the result follows. 

The number of controlled-SWAPs and $T$-count is immediate from the analysis of~\Cref{alg:weak-Schur-Sampling}. 
\end{proof}
 \section{Mixed versus pure one-way state generators via anti-symmetrization}

This section presents an illustrative example application of anti-symmetrization to quantum cryptography. 
To preface our application, we begin by describing one-way state generators. 
One-way state generators were created to be a quantum analog to one-way functions, and roughly are states that can be efficiently prepared from a key $k$, can be efficiently verified, given copies of the state and key, and for whom the key $k$ cannot be efficiently learned. 

In their original definition~\cite{morimae2024one}, one-way state generators were assumed to have a \emph{computationally} efficient verification algorithm, and indeed if the output of a one-way state generator is pure, it always has a canonical, computationally efficient verification (namely, prepare the state $\ket{\psi_k}$ and perform a SWAP test). 
However, when the rank of the one-way state is allowed to be greater than one, the resulting states might no longer admit computationally efficient verification (although they will always admit an inefficient verification).
This led to the definition of \emph{inefficiently verifiable} one-way state generators~\cite{batra2024commitments}, which were shown to be equivalent to EFI pairs~\cite{brakerski2022computational}.
In a series of works~\cite{bostanci2025oracle, behera2024new, goldin2025translating}, it was shown that these two notions (efficiently and inefficiently verifiable) of one-way state generators can be separated relative to a state and unitary oracle.  However, it is an open question whether efficiently verifiable mixed-state one-way state generators are equivalent to pure-state one-way state generators.  

Using the anti-symmetrization algorithm, we can show: when the rank of a one-way state generator is promised to be polynomial in the security parameter $\lambda$, there is a reduction from mixed one-way state generators to pure state output one-way state generators.  
Our result requires a slight modification to the definition of one-way state generators, namely adding the assumption that there is a non-negligible lower bound on the smallest non-zero eigenvalue . 
We formally define gapped one-way state generators below.  

\begin{definition}[Gapped rank-$r$ one-way state generators]
    A one-way state generator, described by algorithms $(\mathsf{KeyGen}, \mathsf{StateGen}, \mathsf{Ver})$, is an \emph{$\epsilon$-gapped rank-$r$ one-way state generator} if for all $k$, $\rho_{k} \gets \mathsf{StateGen}(k)$ is rank-$r$ and the smallest non-zero eigenvalue of $\rho_{k}$ is at least $\epsilon$.  
\end{definition}

We prove the following corollary of our anti-symmetrization result. 
\begin{corollary}
    Pure-state output one-way state generators exist if and only if there exist polynomials $r$, $t$ such that $1/t$-gapped rank-$r$ one-way state generators exist.   
\end{corollary}
\begin{proof}
    In one direction, pure-state output one-way state generators are already $1$-gapped rank-$1$ one-way state generators, so their existence immediately implies the existence of mixed state output one-way state generators.

In the other direction, say that $1/t$-gapped rank-$r$ one-way state generators exist, and for key $k$, say that the one-way state generator outputs a state $\rho_k$ that has spectral decomposition
\begin{equation*}
    \rho_k = \sum_{i} \lambda_i \ketbra{\psi_{k, i}}\,.
\end{equation*}
Then we claim that the pure-state output one-way state generator that maps the key $k$ to the state
\begin{equation*}
    \ket{\psi_k} \propto \sum_{\sigma \in S_{r}} \mathrm{sign}(\sigma) P_{\sigma} \left(\ket{\psi_{k, 1}} \otimes \ldots \otimes \ket{\psi_{k, r}}\right)
\end{equation*}
exists.
First, we note that our anti-symmetrization procedure in (\Cref{alg:sample-inefficient-antisym}) prepares, except with negligible failure probability, the pure state $\ket{\psi_k}$ from polynomially many copies of $\rho_k$ in $\mathrm{poly}(r, t, \log(d))$ time, since $\rho_k$ is assumed to have rank $r$ and its minimum eigenvalue is assumed to be at least $1/t$.  It is also clear that an algorithm that could break the one-wayness of this one-way state generator can break the one-wayness of the original one-way state generator by first applying the anti-symmetrization algorithm (on polynomially many copies of $\rho_k$) and running the algorithm for breaking the one-wayness of the pure-state one-way state generator.  
\end{proof}

Finally, we note that this reduction is ``almost-tight'' among black-box reductions between mixed-state and pure-state one-way state generators, in the sense that when $r$ becomes sufficiently super-polynomial (growing faster than $\lambda^{\log(\lambda)}$), there is a quantum state separation between rank $r$ one-way state generators and efficiently verifiable one-way state generators, using the work of \cite{bostanci2025oracle, behera2024new} and \cite{batra2024commitments}, which can be lifted to a unitary oracle separation by using \cite{goldin2025translating}.  
\begin{claim}
    For every function $r(\lambda) \leq 2^{\lambda}$ that satisfies $\log(r) = \omega(\log^2(\lambda))$, there exists a quantum state oracle (or unitary oracle) relative to which $(1/r)$-gapped rank-$r$ one-way state generators exist, but efficiently verifiable one-way state generators (and thus in particular pure-state output one-way state generators) do not exist.
\end{claim}
\begin{proof}
    Fix a local dimension $d(\lambda) = 2^{\sqrt{\log(r(\lambda))}}$.  For $r$ as in the claim, this is super-polynomial.  The oracle is (for every $\lambda$) to take a Haar random state in $d$ dimensions, as well as a $\mathsf{UnitaryPSPACE}$ oracle.  From \cite{bostanci2025oracle}, efficiently verifiable one-way state generators do not exist relative to this oracle.  Additionally, from the work of \cite{ananth2024cryptography, chen2025power}, there exists a construction of 1PRS relative to this oracle that have dimension $d$, and such that no adversary receiving $t(\lambda)$ copies of the state can distinguish between the construction and a maximally mixed state on $d$-dimensions with probability greater than $O(t^2 / d) = \negl(\lambda)$.  Recall that the construction of one-way state generators from \cite{batra2024commitments} is to fix a negligible $\eta = \eta(\lambda)$ (which determines the soundness of the construction) and map a key of $\log(1/\eta)$ bits to the same number of copies of either the 1PRS or the maximally mixed state on $d$ dimensions.  Taking $\eta = 1/d$ (which is negligible), we have that the total dimension of the one-way state generator is $d^{\log(d)} = r$, and an adversary can only invert the one-way state generator with negligible probability in $\lambda$.  

    To prove a bound on the minimum non-zero eigenvalue of this one-way state generator, we have to inspect the construction of \cite{ananth2024cryptography}.  The construction of EFI pairs implied by \cite{ananth2024cryptography} involves either outputting a maximally mixed state, or outputting $\mathbb{E}_{k}[Z^{k} \ketbra{\psi} Z^{k}]$, where the $Z$ operators act on the first $m$ qubits of the system (where $m$ can be chosen to be anything $\omega(\log(\lambda))$).  Writing the Haar random state as $\ket{\psi} = \sum_{x} \ket{x}\ket{\psi_{x}}$, this state is exactly $\sum_{x} \ketbra{x} \otimes \ketbra{\psi_{x}}$.  This matrix, then, is block-diagonal, and the minimum eigenvalue is exactly the minimum over strings $x \in \{0, 1\}^{m}$ of $\Tr[(\ketbra{x} \otimes \mathrm{id}) \ketbra{\psi}]$.  For every fixed $x$, Levy's lemma~\cite{meckes2019random} tells us that the probability that this is less than $2^{-m-1}$ is at most $\exp(-2^{n-m} / 24)$.  Taking a union bound over all $x \in \{0, 1\}^{m}$, when $m$ is, say $\log(d)/2$ (which satisfies the required condition on $m$), tells us that with overwhelming probability the minimum eigenvalue of this state is greater than $1/d$.  On the other hand, the maximally mixed state over $d$-dimensions has a minimum non-zero eigenvalue of $1/d$.  
 
    Thus the minimum eigenvalue of both states in the EFI pair is greater than $1/d$ with overwhelming probability.  Taking the tensor product of $\log(d)$ of them yields a state whose rank is at most $d^{\log(d)} = r$, and whose minimum eigenvalue is at least $1/r$, as desired.
\end{proof}

We note that it remains an open question whether one can get pure-state one-way state generators from very high rank mixed-state one-way state generators, using their efficient verification algorithm, or if the two notions might actually be separate.  We point out that finding an oracle relative to which one-way puzzles do not exist, but efficiently verifiable one-way state generators exist (extending the result of \cite{mantri2026efi}) would separate pure-state one-way state generators from efficiently verifiable mixed-state one-way state generators.   
\newpage
\sloppy
\printbibliography

\appendix
\crefalias{section}{appendix}
\section{Additional representation-theoretic details and an alternate perspective on~\Cref{alg:unitary-Schur-Sampling,alg:quantum_purity_amplification}}
\label{app:Additional_rep_theory_proofs}

This appendix contains some supplementary proofs and technical discussion which may support a reader's understanding of our paper, but which is not required to follow its main narrative.

We remind the reader of the following notation: 
\begin{itemize}
    \item $V_\lambda: $ The irreducible $S_n$ module (Specht module) indexed by $\lambda$.  
    \item $W_\lambda:$ The $\operatorname{GL}
    _d$ module indexed by $\lambda$. 
    \item $D_\lambda:$ The irreducible representation $D_\lambda: \mathbb{C}[S_n] \rightarrow \mathcal{L}(V_\lambda)$ mapping elements of the symmetric group algebra to their action on the $V_\lambda$ module. 
    \item $\Pi^\lambda :$ The projector onto the $\lambda$-isotypic component of $\left(\mathbb{C}^d\right)^{\otimes n}$.
    \item $A(S):$ The element of $\mathbb{C}[S_n]$ which projects onto the anti-symmetric subspace on the indices indicated by $S \subseteq [n]$.
    \item $A(x_1,x_2,...,x_t):$ The element of $\mathbb{C}[S_n]$ which simultaneously projects onto the anti-symmetric subspaces on the first $x_1$ elements of $[n]$, the subsequent $x_2$ elements, and so on. 
\end{itemize}

We will also develop some additional notation to discuss restrictions of Specht modules. For any $n$ and subset $R \subseteq [n]$, let $S_R \times S_{R^c} \subseteq S_n$ be the subgroup of the symmetric group $S_n$ containing permutations which do not mix the elements in $R$ and those in $R^c$. 
Additionally, let $\text{Res}^{S_n}_{S_R, S_{R^c}}$ be the functor which maps the Specht module $V_\lambda$ to the module with the same vector space and representation, but where the group is restricted to $S_R \times S_{R^c}$. This notation lets us restate Pieri's rule, now without any reference to states in $\left(\mathbb{C}^d\right)^{\otimes n}$.

\begin{lemma}[Pieri's Rule\footnote{Technically only the second half of this statement (concerning $c^\lambda_{\mu,\nu}$ when $\mu$ is a single column partition) is Pieri's rule, while the first half sets up the more general Littlewood-Richardson rule. }]
    For any $n$, partition $\lambda \vdash n$, and  $R\subseteq [n]$ there is an isomorphism between modules: 
    \begin{align}
        \text{Res}^{S_n}_{S_R, S_{R^c}} V_\lambda \cong \bigoplus_{\mu \, \vdash |R|, \,\nu \, \vdash|R^c|} c^\lambda_{\mu,\nu} \left(V_\mu \otimes V_\nu\right).
    \end{align}
    Equivalently, there is a $(S_R \times S_{R^c})$-equivariant homomorphism mapping vector spaces
    \begin{align}
        \phi^\lambda_{R,R^c}:V_\lambda \rightarrow \bigoplus_{\mu \, \vdash |R|, \,\nu \, \vdash|R^c|}  \left(V_\mu \otimes V_\nu\right)^{\oplus c^\lambda_{\mu,\nu}}.
    \end{align}
If $\mu = (1,1,...,1)$ is a one column representation then we have 
    $
        c^\lambda_{\mu,\nu} =1
    $ whenever $\nu$ can be obtained from $\lambda$ by deleting exactly one box from $|R|$ distinct rows of $\lambda$ and 
    $
        c^\lambda_{\mu,\nu} =0
    $
    otherwise. 
\end{lemma}

We are now in a position to prove a (standard) technical statement which is used without proof in the main document. 

\begin{lemma}
    \label{lem:one_dimensional}
    For every partition $\lambda$, there is a unique element $\ket{0_\lambda} \in V_\lambda$ with 
    \begin{align}
        D_\lambda(A(\lambda_1', \lambda_2',...,\lambda_{\lambda_1}'))\ket{0_\lambda} = \ket{0_\lambda}.
    \end{align}
    Equivalently, for each partition $\lambda$, the matrix $D_\lambda(A(\lambda_1', \lambda_2',...,\lambda_{\lambda_1}'))$ is a rank one projector. 
\end{lemma}

\begin{proof}
    We prove the result by induction on the size of $\lambda_1$, i.e. the number of columns of $\lambda$. The base case is immediate, since any one column irrep corresponds to the sign representation, and in particular is one dimensional. 

    Next we consider an arbitrary partition $\lambda$.
    We take $R= [\lambda_1'] \subseteq [n]$ and consider the restriction $\text{Res}^{S_n}_{S_R, S_{R^c}}$ applied to $V_\lambda$, with $\phi^\lambda_{R, R^c}$ the corresponding homomorphism. By equivariance, we have 
    \begin{align}
        \phi^\lambda_{R, R^c} \left(D_\lambda(A(\lambda_1', \lambda_2',...,\lambda_{\lambda_1}'))\right) = \bigoplus_{\mu \, \vdash |R|, \,\nu \, \vdash|R^c|} \left(D_\mu(A(\lambda_1'))  \otimes D_\nu (A(\lambda_2',\lambda_3',...,\lambda_{\lambda_1}'))\right)^{c^\lambda_{\mu, \nu}}\phi^\lambda_{R, R^c}.
    \end{align}
    Because $A(\lambda_1')$ is the projector onto the anti-symmetric subspace,
    $
        D_\mu(A(\lambda_1'))
    $
    is nonzero only when $\mu$ is the single partition $(1,1,...,1)$. Note this column has length $|R| = \lambda_1'$. Fix this value for $\mu$. But then Pieri gives that $c^{\lambda}_{\mu,\nu}$ is nonzero only when $\nu$ is the unique partition which can be obtained by deleting one box from each row of $\lambda$. Fix this value for $\nu$. Then we have 
        \begin{align}
        \phi^\lambda_{R, R^c} \left(D_\lambda(A(\lambda_1', \lambda_2',...,\lambda_{\lambda_1}'))\right) = \left(D_\mu(A(\lambda_1'))  \otimes D_\nu (A(\lambda_2',\lambda_3',...,\lambda_{\lambda_1}'))\right)\phi^\lambda_{R, R^c}.
    \end{align}
    But $D_\mu(A(\lambda_1'))$ is the identity on a one-dimensional space, and $ D_\nu (A(\lambda_2',\lambda_3',...,\lambda_{\lambda_1}'))$ is a projector onto a one dimensional space by our induction hypothesis. Since $\phi^\lambda_{R, R^c}$ is a homomorphism, we are done. 
\end{proof}

We now comment briefly on how~\Cref{lem:one_dimensional} can be used to understand the output of~\Cref{alg:unitary-Schur-Sampling,alg:quantum_purity_amplification}, providing an alternate perspective on those algorithms. First observe that, for any partition $\lambda$, 
\begin{align}
    &\Uschur \left(\Pi^{\lambda_1'}_{\mathrm{anti}} \otimes \Pi^{\lambda_2'}_{\mathrm{anti}} \otimes ... \otimes \Pi^{\lambda_{\lambda_1}'}_{\mathrm{anti}}\right) \Pi^\lambda \Uschur^\dagger 
= D_\lambda(A(\lambda_1',\lambda_2',\lambda_3',...,\lambda_{\lambda_1}')) \otimes I_{W_\lambda}
\end{align}
and we have just shown that $D_\lambda(A(\lambda_1',\lambda_2',\lambda_3',...,\lambda_{\lambda_1}'))$ is a projector onto a one dimensional space. Then we have, for any $\sigma \in S_n$
\begin{align}
    \text{Im} \left(\left(\Pi^{\lambda_1'}_{\mathrm{anti}} \otimes \Pi^{\lambda_2'}_{\mathrm{anti}} \otimes ... \otimes \Pi^{\lambda_{\lambda_1}'}_{\mathrm{anti}}\right) \Pi^\lambda  \right) \supseteq \text{Im}\left(\left(\Pi^{\lambda_1'}_{\mathrm{anti}} \otimes \Pi^{\lambda_2'}_{\mathrm{anti}} \otimes ... \otimes \Pi^{\lambda_{\lambda_1}'}_{\mathrm{anti}}\right) P_\sigma \, \Pi^\lambda  \right) 
\end{align}
with equality whenever 
\begin{align}
    \left(\Pi^{\lambda_1'}_{\mathrm{anti}} \otimes \Pi^{\lambda_2'}_{\mathrm{anti}} \otimes ... \otimes \Pi^{\lambda_{\lambda_1}'}_{\mathrm{anti}}\right) P_\sigma \, \Pi^\lambda \neq 0.
\end{align}
A similar equation holds when $P_\sigma$ is replaced by the natural representation of any polynomial in $\mathbb{C}[S_n]$. 

Now consider the output of~\Cref{alg:unitary-Schur-Sampling} applied to some arbitrary input state $\rho^{\otimes n} \in \mathcal{L}\left((\mathbb{C}^d )^{\otimes n}\right)$. Conditioned on its returning partition label $\lambda$, and ignoring the $\delta/2$ probability of error, we can write the output of the algorithm (up to normalization) as 
\begin{align}
    \sum_r\Pi^\lambda \left(\Pi^{\lambda_1'}_{\mathrm{anti}} \otimes \Pi^{\lambda_2'}_{\mathrm{anti}} \otimes ... \otimes \Pi^{\lambda_{\lambda_1}'}_{\mathrm{anti}}\right)   K_r \rho^{\otimes n} K_r^\dagger \left(\Pi^{\lambda_1'}_{\mathrm{anti}} \otimes \Pi^{\lambda_2'}_{\mathrm{anti}} \otimes ... \otimes \Pi^{\lambda_{\lambda_1}'}_{\mathrm{anti}}\right) \Pi^\lambda 
\end{align}
where the $\Pi^\lambda$ and $\left(\Pi^{\lambda_1'}_{\mathrm{anti}} \otimes \Pi^{\lambda_2'}_{\mathrm{anti}} \otimes ... \otimes \Pi^{\lambda_{\lambda_1}'}_{\mathrm{anti}}\right)$ projectors express the fact the output state is a column-anti-symmetric state in the $\lambda$ isotypic component, and the $K_r$ are Kraus operators describing the full sequences of anti-symmetrization measurements (i.e. applications of \Cref{alg:anti-sym-growing}) and subsequent rearrangements of qudits which occurred during the algorithm. The sum over $r$ expresses the fact there are many possible records of intermediate measurement outcomes which can lead to~\Cref{alg:unitary-Schur-Sampling} outputting partition label $\lambda$, and the algorithm deliberately forgets those. But critically, each $K_r$ is just a representation of some polynomial in $\mathbb{C}(S_n)$, so 
\begin{align}
    \text{Im}\left(\Pi^\lambda \left(\Pi^{\lambda_1'}_{\mathrm{anti}} \otimes \Pi^{\lambda_2'}_{\mathrm{anti}} \otimes ... \otimes \Pi^{\lambda_{\lambda_1}'}_{\mathrm{anti}}\right)   K_r \right)
\end{align}
is the same, independent of $K_r$. This gives an alternate explanation of why no record of the sequence of anti-symmetrization measurements other than the lengths of the final anti-symmetric states produced appears in the permutation information (i.e. the $V_\lambda$ module of the Schur-transformed state) of the state produced by~\Cref{alg:unitary-Schur-Sampling}.

We can use the same idea to compare~\Cref{alg:quantum_purity_amplification} to the original purity amplification algorithm given in~\cite{li2024optimal}. The first two steps of that algorithm measure the isotypic component and YY basis information of an input $\rho^{\otimes n} \in \left(\mathbb{C}^d\right)^{\otimes n}$, then apply a measurement-dependent correction to rotate to the column-ordered YY basis state, which (following their notation) we denote $\ket{\lambda^{\diamond}}$. Conditioned on the measuring the $\lambda$-isotypic component we can write the post-measurement state produced by this protocol as
\begin{align}
    \sum_{\ket{T} \in YY(V_\lambda)} \Uschur^\dagger (\ketbra{\lambda^\diamond}{T} \otimes I_{W_\lambda}) \Uschur  \Pi^\lambda \rho^{\otimes n} \Pi^\lambda \Uschur^\dagger (\ketbra{T}{\lambda^\diamond} \otimes I_{W_\lambda}) \Uschur
\end{align}
where the sum runs over the set of all YY basis vectors for $V_\lambda$, which we denote $YY(V_\lambda)$ and which are described by Young Tableau. We now make a basic, but critical, observation about the column ordered YY basis element $\ket{\lambda^\diamond}$
\begin{lemma}
    \label{lem:YY_column_anti}
    For any partition $\lambda \vdash n$ let $\ket{\lambda^\diamond}$ be the column-ordered YY basis element. Then we have 
    \begin{align}
        D_\lambda(A(\lambda_1', \lambda_2',...,\lambda_{\lambda_1}') ) \ket{\lambda^\diamond} = \ket{\lambda^\diamond} 
    \end{align}
\end{lemma}
\begin{proof}
    By Young's orthogonal form (see, for example ~\cite[Section 3.4]{james2006representation}) for any YY basis vector indexed by Young Tableau T we have
    \begin{align}
        D_\lambda((i, i+1))\ket{T} = -\ket{T}
    \end{align}
    whenever symbols $i$ and $i+1$ appear in the same column of T. The result follows.
\end{proof}
Combining \Cref{lem:YY_column_anti,lem:one_dimensional} gives that $\ket{\lambda^\diamond} = \ket{0_\lambda}$, which is the technical fact used in the proof of~\Cref{thm:quantum_purity_amplification}. It also gives, for any $\ket{T} \in YY(V_\lambda)$ 
\begin{align}
    \Uschur^\dagger (\ketbra{\lambda^\diamond}{T} \otimes I_{W_\lambda}) \Uschur  \Pi^\lambda = \left(\Pi^{\lambda_1'}_{\mathrm{anti}} \otimes \Pi^{\lambda_2'}_{\mathrm{anti}} \otimes ... \otimes \Pi^{\lambda_{\lambda_1}'}_{\mathrm{anti}}\right)\Uschur^\dagger (\ketbra{\lambda^\diamond}{T} \otimes I_{W_\lambda}) \Uschur  \Pi^\lambda
\end{align}
and  
\begin{align}
    \text{Im}(\Uschur^\dagger (\ketbra{\lambda^\diamond}{T} \otimes I_{W_\lambda}) \Uschur  \Pi^\lambda) = \text{Im}\left(\left(\Pi^{\lambda_1'}_{\mathrm{anti}} \otimes \Pi^{\lambda_2'}_{\mathrm{anti}} \otimes ... \otimes \Pi^{\lambda_{\lambda_1}'}_{\mathrm{anti}}\right) \Pi^\lambda\right).
\end{align}
So we see the image of the first two steps of~\cite{li2024optimal} is the same as the image of~\Cref{alg:unitary-Schur-Sampling}, giving an intuitive justification for~\Cref{alg:quantum_purity_amplification}. 
\end{document}